\documentclass{article}

\usepackage{PRIMEarxiv}

\usepackage[utf8]{inputenc} % allow utf-8 input
\usepackage[T1]{fontenc}    % use 8-bit T1 fonts
\usepackage{hyperref}       % hyperlinks
\usepackage{url}            % simple URL typesetting
\usepackage{booktabs}
\usepackage{multirow}       % professional-quality tables
\usepackage{amsfonts}       % blackboard math symbols
\usepackage{nicefrac}       % compact symbols for 1/2, etc.
\usepackage{microtype}      % microtypography
\usepackage{fancyhdr}       % header
\usepackage{graphicx}       % graphics
\graphicspath{{media/}}     % organize your images and other figures under media/ folder

\usepackage{amsmath,mathtools,amssymb}

\usepackage[ruled,vlined,linesnumbered]{algorithm2e} % main package
\usepackage{tikz}
\usetikzlibrary{arrows.meta, positioning, calc,fit,backgrounds,shadows}
\usepackage[numbers]{natbib}

\newtheorem{theorem}{Theorem}
\newtheorem{lemma}[theorem]{Lemma}
\newtheorem{proof}{Proof}
\newtheorem{definition}{Definition}

\newtheorem{proposition}{Proposition}
\newtheorem{corollary}{Corollary}

\title{Adaptive Accountability in Networked MAS:\\
Tracing and Mitigating Emergent Norms at Scale
\thanks{An early version of this paper was published in AAAI/ACM AIES 2025~\cite{Alqithami_2025}: \url{https://doi.org/10.1609/aies.v8i1.36536}} 
}

\author{
  Saad Alqithami\\%\thanks{Corresponding author.} \\
  Computer Science Department, Al-Baha University, Al-Baha, Saudi Arabia \\
  \texttt{salqithami@bu.edu.sa} \\
}

\begin{document}
\maketitle

\begin{abstract}
Large-scale networked multi-agent systems increasingly underpin critical infrastructure, yet their collective behavior can drift toward undesirable emergent norms such as collusion, resource hoarding, and implicit unfairness.
We present the \emph{Adaptive Accountability Framework} (AAF), an end-to-end runtime layer that (i) records cryptographically verifiable interaction provenance, (ii) detects distributional change points in streaming traces, (iii) attributes responsibility via a causal influence graph, and (iv) applies cost-bounded interventions---reward shaping and targeted policy patching---to steer the system back toward compliant behavior.
We establish a bounded-compromise guarantee: if the expected cost of intervention exceeds an adversary's expected payoff, the long-run fraction of compromised interactions converges to a value strictly below one.
We evaluate AAF in a large-scale factorial simulation suite (\(87{,}480\) runs across two tasks; up to 100 agents plus a 500-agent scaling sweep; full and partial observability; Byzantine rates up to \(10\%\); 10 seeds per regime).
Across 324 regimes, AAF lowers the \emph{executed} compromise ratio relative to a Proximal Policy Optimization baseline in 96\% of regimes (median relative reduction 11.9\%) while preserving social welfare (median change 0.4\%).
Under adversarial injections, AAF detects norm violations with a median delay of 71 steps (interquartile range 39--177) and achieves a mean top-ranked attribution accuracy of 0.97 at 10\% Byzantine rate.
\end{abstract}

% keywords can be removed

\section{Introduction}
\label{sec:introduction}

Multi-agent systems (MAS) increasingly underpin critical applications in transportation, smart-grid energy management, finance, and healthcare~\citep{shoham2007if,wooldridge2009introduction}.  
By distributing decision making across many partly autonomous agents, these systems offer scalability, resilience to single-point failures, and rapid adaptation to non-stationary environments.  
Those benefits, however, come with \emph{interaction complexity}: once deployed, the collective behavior can drift toward undesirable emergent norms that were never explicitly designed or anticipated~\citep{stone2000multiagent,Rahwan2019MachineBehavior}.  
Identifying and correcting such norms is essential for ensuring fairness, security, and compliance with societal values.

In large-scale networked MAS, responsibility is diffused across agents and time steps. Classical AI-accountability frameworks—often built for single models or clearly identifiable human decision makers—rely on static compliance checks or post-hoc audits~\citep{mittelstadt2019principles,jobin2019global}. These approaches break down when no single node sees the full global state and harms can emerge from subtle feedback loops. Recent policy instruments, such as the EU Artificial Intelligence Act~\cite{eu_ai_act_2024} and the U.S.\ NIST AI Risk-Management Framework~\cite{nist_ai_rmf_2023}, call for ``collective accountability'' in AI infrastructures but leave open the technical challenge of tracing and mitigating harmful norms online and in a decentralized fashion.

Three established research threads touch on this challenge yet remain individually insufficient. 
\emph{Normative MAS} studies formal norms and sanctions~\citep{Andrighetto2013NormDynamics,Hollander2011NormativeSystems} but typically assumes a central monitor with perfect state access.  
\emph{Multi-agent reinforcement learning} (MARL) has produced sophisticated training algorithms~\citep{busoniu2008comprehensive,Foerster2018COMA}, but only recently has begun to explore equity, collusion, or convergence to harmful equilibria.  
\emph{Runtime assurance} for cyber-physical systems proposes supervisory safety guards~\citep{perez2024artificial}, yet seldom scales beyond a handful of agents or supports changing objectives.  
Consequently, the field still lacks a comprehensive method to \emph{detect}, \emph{trace}, and \emph{correct} undesirable emergent norms under partial observability and heterogeneous incentives.

This paper tackles four questions:  
(i)~How can responsibility for distributed actions be continuously attributed when no participant sees the entire trajectory?  
(ii)~Can harmful emergent norms be detected online without halting operations or requiring global state?  
(iii)~Which local policy or reward interventions can steer a large MAS back toward socially preferred outcomes while respecting resource and latency constraints?  
(iv)~Do those interventions remain robust as the number of agents, communication topology, or payoff structure changes?

We answer by proposing an \emph{Adaptive Accountability Framework} that combines a lifecycle-aware audit ledger, decentralized sequential hypothesis tests for online norm detection, and targeted reward-shaping or policy-patch interventions. Our principal theoretical result—the \emph{bounded-compromise theorem}—proves that when the expected cost of interventions exceeds an adversary’s payoff, the steady-state fraction of compromised or collusive interactions converges to a value strictly below one. The theorem formalizes the intuition that inexpensive, well-targeted corrections can prevent persistent harmful norms even in strategic environments.

To evaluate the framework at realistic scale, we run a Q1-scale factorial simulation suite comprising \(87{,}480\) episodes (2 tasks; \(N\le 100\) (plus a scaling sweep to \(N=500\) in \texttt{resource\_sharing}); full and partial observability; Byzantine rates up to \(10\%\); 10 seeds per regime).
Across 324 regimes, AAF reduces the executed compromise ratio relative to PPO-only in 96\% of regimes (median reduction 11.9\%; \(p<0.05\) in 86.4\%) while preserving social welfare (median change 0.4\%).
Inequality effects are mixed, highlighting a tunable compromise--fairness trade-off (see §\ref{subsec:results}).
All code and simulation scripts will be released to support replication.

The paper has the following main contributions:
\begin{enumerate}
    \item An accountability architecture that logs interaction events, tags causal chains, and maps evolving responsibility flows without centralized state.
    \item Online norm-detection algorithms based on decentralized sequential hypothesis testing.
    \item Adaptive mitigation protocols that apply graded local incentives or policy patches to realign global behavior.
    \item A high-fidelity implementation integrated with modern MLOps pipelines and illustrated through resource-allocation and collaborative-task case studies.
    \item Comprehensive large-scale evaluation demonstrating robustness and scalability across diverse operating conditions.
\end{enumerate}

The rest of the paper is organized as follows. 
Section~\ref{sec:background} surveys related work on MAS coordination, normative design, and AI governance.  
Section~\ref{sec:problem-definition} formalizes the problem setting.  
Section~\ref{sec:proposed-framework} details the proposed architecture, detection tests, and intervention algorithms.  
Section~\ref{sec:implementation} describes the software stack and deployment considerations.  
Section~\ref{sec:experiments} reports empirical results, and Section~\ref{sec:discussion} discusses their implications and limitations.  
Finally, Section~\ref{sec:conclusion-future-work} concludes with directions for future research in ethically aligned MAS.

\section{Background and Literature Review}
\label{sec:background}

Multi-agent systems consist of multiple autonomous entities that sense, decide, and act on local information while exchanging limited signals with neighbors \citep{wooldridge2009introduction,shoham2007if}.  Delegating control in this way yields scalability, fault tolerance, and graceful adaptation, enabling real-world deployments in smart-grid scheduling \citep{khan2019multiagent}, adaptive traffic control \citep{yan2020reinforcement}, supply-chain optimization \citep{giannakis2020multi}, cooperative robotics \citep{brambilla2013swarm}, distributed optimization on sensor networks \citep{yang2019survey}, algorithmic trading \citep{tesfatsion2006agent}, and cross-hospital resource sharing \citep{herrouz2019multi}.  What unites these domains is the need to reconcile \emph{local incentives} with \emph{global constraints} under partial observability and stringent latency budgets.

Nonetheless, the same decentralization that confers robustness breeds \emph{interaction complexity}. Agents are autonomous, social, and both reactive and proactive \citep{wooldridge1995intelligent}; even simple local rules can induce system-level phenomena that are hard to predict, verify, or control \citep{stone2000multiagent}.  Coordination grows more challenging as populations scale, links fluctuate, or objectives diverge \citep{corkill2003collaborating}.  Emergent behaviors—benign or harmful—complicate formal safety guarantees in partially observable settings \citep{busoniu2008comprehensive}.

%---------------------------------------------------------------------
\subsection{Emergent Norms: From Descriptive Theory to Algorithmic Design}
%---------------------------------------------------------------------

Classical emergence studies viewed norms as regularities not explicit in any single agent's code \citep{bedau2002downward}.  Early mechanisms included imitation and social learning \citep{sen2007emergence}, reinforcement-driven sanctioning \citep{Morries2019norm}, and top-down institutions with explicit penalties \citep{conte2012normative,pitt2012axiomatization}.  Recent work shifts from \emph{descriptive} to \emph{constructive}: \citet{Tzeng2024NormEnforcement} show that ``soft-touch'' speech acts (tell, hint) enforce norms faster than pure sanction; \citet{Oldenburg2024SharedNorms} introduce Bayesian rule-induction so newcomers infer latent norms on-line; \citet{Serramia2024Consensus} formalize optimal consensus norms that balance heterogeneous stakeholder utilities; \citet{Ren2024CRSEC} harness large language models to generate conflict-reducing social norms in an artificial town; and \citet{Woodgate2025RawlsianNorms} operationalization Rawls' maximin principle, steering emergent norms toward lower inequality. Together these papers recast norm formation as an algorithmic and ethical design problem.

%---------------------------------------------------------------------
\subsection{Distributed Accountability and Auditing}
%---------------------------------------------------------------------

Standard artifacts—model cards \citep{mitchell2019model} and fact sheets \citep{arnold2019factsheets}—target monolithic models. Applied to MAS they fail to address diffuse responsibility \citep{rastogi2022distributed}, continual policy drift \citep{jobin2019global}, and aggregation of local harms. The recent explosion of Large Language Model (LLM)-based multi-agent systems \citep{Li2024SurveyLLMMAS} has further exacerbated these challenges, as generative agents exhibit complex emergent behaviors that are difficult to trace. 

Recent research begins to fill this gap: \citet{Chan2024Visibility} catalogue visibility tools for peer-to-peer deployments, trading off informativeness and privacy; \citet{Mu2024Responsibility} extend alternating-time temporal logic with quantitative responsibility metrics; \citet{Gyevnar2024CEMA} generate natural-language causal explanations via counterfactual roll-outs, boosting trust in autonomous vehicles; and \citet{Chang2025Chronicles} treat provenance as a chain of custody for LLM-generated artifacts. \citet{Ebrahimi2025Survey} systematically survey LLM-based MAS through the ``pillars of responsibility'' (reliability, transparency, accountability, and fairness), while \citet{Raza2025TRiSM} adapt the Trust, Risk, and Security Management (TRiSM) framework specifically for agentic AI. Furthermore, \citet{Solomon2025LumiMAS} propose LumiMAS, a comprehensive framework for real-time monitoring and enhanced observability in multi-agent systems, highlighting the critical need for runtime oversight.

%---------------------------------------------------------------------
\subsection{Intervention and Governance Mechanisms}
%---------------------------------------------------------------------

Detection alone is insufficient; an accountability layer must \emph{intervene}. Beyond classic reward shaping, recent levers include action-space restriction \citep{Oesterle2024RAISE}, monitored MDPs with human oracles \citep{Parisi2024MonitoredMDP}, and sample-efficient opponent shaping \citep{Fung2024OpponentShaping}. In the realm of safe MARL, constrained policy optimization and Lagrangian relaxation are widely used to enforce safety constraints \citep{Achiam2017CPO}. Similarly, \citet{Li2024Byzantine} explore Byzantine robust cooperative MARL, formulating the problem as a Bayesian game to maintain coordination even when a fraction of agents act adversarially. \citet{Zheng2025Rethinking} further rethink MAS reliability from the perspective of Byzantine fault tolerance, emphasizing the need for robust consensus mechanisms in agentic systems. Model-based credit assignment with counterfactual imagination \citep{Chai2024MACD} shows promise for scalable cooperation—although its computational cost remains high in large MAS.

%---------------------------------------------------------------------
\subsection{Causal Tracing and Lifecycle Provenance}
%---------------------------------------------------------------------

Effective intervention presupposes causal insight. Counterfactual credit assignment has long been used in MARL (e.g., COMA) \citep{Foerster2018COMA}, and more recent work extends causal effect decomposition to multi-agent sequential decision making \citep{Triantafyllou2024Counterfactual}. In offline settings, \citet{Wang2024MACCA} propose MACCA, a framework that assigns credit by analyzing causal relationships in individual rewards. Time-uniform concentration bounds \citep{howard2020time} and stochastic approximation theory \citep{kushner2003stochastic} provide rigorous guarantees for sequential detectors and adaptive thresholds such as those used in our framework.

%---------------------------------------------------------------------
\subsection{Regulatory Impetus and Remaining Gaps}
%---------------------------------------------------------------------

The EU AI Act and the NIST AI RMF (2023) explicitly flag collective emergent risk, yet supply no technical recipe for tracing accountability in networked AI. Recent literature contributes building blocks—richer norm learning, quantitative responsibility metrics, scalable intervention levers—but lacks an online, end-to-end pipeline that meets bandwidth budgets and supplies audit-grade evidence. Three open needs persist:

\begin{enumerate}
\item Continuous monitoring attuned to policy drift and lossy observations;
\item Distributed attribution robust to adversarial spoofing;
\item Low-latency interventions that safeguard global welfare while respecting local autonomy.
\end{enumerate}

%---------------------------------------------------------------------
\subsection{Positioning of the Present Work}
%---------------------------------------------------------------------

The Adaptive Accountability Framework introduced next addresses all three needs. It combines a cryptographically grounded Merkle ledger, time-uniform norm detectors, and cost-bounded interventions validated at 100-agent scale (and up to 500 agents in a scaling sweep). The next sections detail the architecture (Section~\ref{sec:proposed-framework}), furnish theoretical guarantees (Section~\ref{sec:analysis}), and supply empirical evidence that the framework curbs harmful emergence even under adversarial disturbance (Section~\ref{sec:experiments}).

\section{Problem Definition and Research Questions}
\label{sec:problem-definition}

This section formalizes the accountability problem in large-scale networked multi-agent systems, frames the research questions that guide our work, and states the operational assumptions under which any proposed solution must function. The notation introduced here flows directly into the analytical guarantees of Section \ref{sec:analysis} and the implementation described in Section \ref{sec:implementation}.

\subsection{Formalizing Accountability in Networked Multi-Agent Systems}
\label{sec:formalisation}

\paragraph{System model:}
We represent the MAS as a discrete-time, partially observable stochastic game
\[
\mathcal{G}=\bigl\langle 
  \mathcal{S},\,\{\mathcal{A}_i\}_{i=1}^{N},\,\{\Omega_i\}_{i=1}^{N},\,
  T,\,\{\mathcal{R}_i\}_{i=1}^{N},\,\gamma
\bigr\rangle .
\]
The global state space $\mathcal{S}$ may include physical variables (e.g.\ queue lengths, robot poses) and latent coordination artifacts (e.g.\ shared plans). At each round $t\in\mathbb{N}$ the environment occupies $s_t\!\in\!\mathcal{S}$. Agent $A_i$ acquires a private observation $o_{i,t}\!\in\!\Omega_i$ parameterized by a perceptual kernel $P(o_{i,t}\!\mid s_t,i)$; selects an action $a_{i,t}\!\in\!\mathcal{A}_i$ via a history-dependent policy $\pi_{i,t}\!\in\!\Delta(\mathcal{A}_i)^{\Omega_i^{t+1}}$; and receives an individual reward
\[
r_{i,t}=\,\mathcal{R}_i(s_t,\mathbf{a}_t)\;=\;r_{i,t}^{\text{private}}\;+\;\lambda\,r_{t}^{\text{social}},
\]
where $\mathbf{a}_t=(a_{1,t},\dots,a_{N,t})$ and $\lambda\!\in\![0,1]$ trade off egoistic and collective incentives. The transition kernel $T$ yields $s_{t+1}\!\sim\!T(s_t,\mathbf{a}_t)$; consequently, learning updates $\pi_{i,t+1}\!\leftarrow\!\mathrm{Learn}(\pi_{i,t},\mathcal{D}_{i,t})$ on local traces $\mathcal{D}_{i,t}$, imparting non-stationarity to the joint dynamics $T\!\circ\!\Pi_t$ \citep{hernandez2019survey}.  

\paragraph{Communication substrate:}
Agents communicate over a time-varying graph $G_t=(V,E_t)$ whose edge set changes with failures or mobility. Each edge $(i,j)$ carries at most $B_{\max}$ bytes per step; messages are subject to random delay $\delta_{i,j}\!\sim\!\mathrm{Geo}(p_{\text{drop}})$ and may be adversarially reordered. We model the aggregate channel as an \textsc{erasure-and-delay} network to capture practical wireless constraints in vehicular or drone swarms.

\paragraph{Event ledger and responsibility flow:}
Accountability demands a tamper-evident record of ``who did what, when, and with what effect''. We therefore maintain a rolling directed acyclic graph
\[
\mathcal{L}_T=\bigl(\mathcal{V}_T,\mathcal{C}_T\bigr),\qquad
\mathcal{V}_T=\!\bigcup_{t=0}^{T}\mathcal{E}_t,
\]
where each vertex $e\in\mathcal{E}_t$ is an atomic event  
$\langle\textsf{id},\,t,\,\mathrm{type},\,\mathrm{payload}\rangle$  (actions, messages, external shocks). An edge $(e_k\!\rightarrow\!e_\ell)\!\in\!\mathcal{C}_T$ encodes empirical \emph{causal influence}. Because full interventional data are unavailable online, we approximate causality with rolling multivariate Granger tests on an $h$-step horizon and retain only edges whose F-statistic exceeds a tunable significance level $\alpha_{\!\text{GC}}$; this yields $\vert\mathcal{C}_T\vert=\mathcal{O}(hN\deg^+(G_t))$ and keeps storage linear in $T$~\cite{sherman1950adjustment}.

To quantify blame we adopt a \emph{Shapley-value-style} responsibility rule
\[
\rho_i(e)=\sum_{\mathcal{S}\subseteq V\setminus\{i\}}
\frac{\lvert\mathcal{S}\rvert!\,(N-\lvert\mathcal{S}\rvert-1)!}{N!}\,
\bigl[\phi(e,\mathcal{S}\cup\{i\})-\phi(e,\mathcal{S})\bigr],
\]
where $\phi(e,\mathcal{S})$ equals $1$ if $e$ would still occur when only agents in $\mathcal{S}$ are allowed to influence upstream events and $0$ otherwise. We compute $\rho_i(e)$ via a Monte-Carlo kernel-SHAP proxy with at most $K$ coalition samples per event, keeping amortized overhead sub-quadratic in $N$.

\paragraph{Norms, violations, and emergent harm:}
System designers declare a finite set $\Phi=\{\phi^{(1)},\ldots,\phi^{(M)}\}$ of norms, each a predicate
\[
\phi^{(m)}:\mathcal{S}\times\mathcal{A}_1\times\!\dots\!\times\mathcal{A}_N
      \longrightarrow \{\textsf{safe},\textsf{violate}\}.
\]
Examples include fairness constraints ($\mathrm{Gini}\!<\!0.3$), no-collusion clauses (\mbox{$Q$-statistic$\!<\!\tau$}), or hard safety guards ($\lVert s_t \!-\! s_{\text{unsafe}}\rVert \!>\! \epsilon$).  
We call a \emph{negative emergent norm} any persistent pattern for which
\[
\Pr_{t\ge t_0}\!\bigl[\phi^{(m)}(s_t,\mathbf{a}_t)=\textsf{violate}\bigr] \;>\;\varepsilon,
\]
even though no individual policy $\pi_i$ encodes $\textsf{violate}$. The accountability objective is twofold:
\[
\min_{\{\sigma_t\}} \;
\limsup_{T\to\infty}\frac{C_T}{T}
\quad\text{s.t.}\quad
\Delta J_{\text{soc}}\le\delta_{\text{perf}},\;
\mathbb{E}\bigl[\text{Cost}(\sigma_t)\bigr]\le c_{\max},
\]
where $C_T$ counts norm-violating events (above), $\delta_{\text{perf}}$ bounds acceptable performance loss, and $\sigma_t$ are interventions defined next.

\paragraph{Explicit Threat Model:}
We consider a networked MAS where a subset of agents may act adversarially or selfishly, threatening system-level norms. Specifically, we assume:
\begin{enumerate}
    \item \textbf{Byzantine Agents:} Up to $\kappa$ agents are Byzantine. They possess full knowledge of the system's state and the accountability framework's detection thresholds. They may deviate arbitrarily from their nominal policies $\pi_{i,t}$, collude to trigger false alarms, or attempt to stay just below the CUSUM detection threshold (stealthy attacks). However, they cannot break the cryptographic primitives securing the ledger.
    \item \textbf{Rational Learners:} The remaining $N-\kappa$ agents are rational reward-maximizers. Their individual objective coefficients $\lambda$ need not align with social welfare, creating the possibility of ``lawful but awful'' emergent behaviors (e.g., tragedy of the commons) even without malicious intent.
    \item \textbf{Network Adversary:} The communication substrate may drop up to a fraction $p_{\text{drop}}$ of messages or reorder them, but it cannot forge cryptographic signatures or permanently partition the network. The graph $G_t$ must remain $(\kappa+1)$-vertex-connected on average.
\end{enumerate}

\paragraph{Intervention model:}
A supervisor (possibly distributed) deploys interventions $\sigma_t=\langle\mathcal{I}_t,\textsf{type},\theta_t\rangle$  with
\begin{itemize}
    \item \textsf{type}$\in\{\textsc{reward-shaping},\textsc{policy-patch},\textsc{link-throttle}\}$,
    \item target set $\mathcal{I}_t\subseteq V$ selected by highest aggregate responsibility $\sum_{e\in\mathcal{E}_{t-h:t}}\rho_i(e)$,
    \item parameters $\theta_t$ (e.g.\ shaping potential $\Phi$, patch weights, or edge weights).
\end{itemize}
Interventions incur cost $\text{Cost}(\sigma_t)$ that combines runtime (CPU/GPU), bytes transmitted, and any monetary or regulatory penalty for altering incentives. Section~\ref{sec:analysis} proves the \emph{bounded-compromise theorem}: if the time-average intervention cost satisfies $\mathbb{E}[\text{Cost}(\sigma_t)] \!>\! \mathbb{E}[\text{Adversary Gain}]$, then there exists $\eta^{\star}\!\in\!(0,1)$ for which $\limsup_{T\to\infty} \tfrac{C_T}{T}\!\le\!\eta^{\star}$ almost surely, irrespective of $\kappa$ and graph topology, provided $G_t$ stays $(\kappa\!+\!1)$-vertex-connected on average.

\paragraph{Complexity and storage bounds:}
Let $d_{\max}$ denote the maximum out-degree in $G_t$ and $h$ the causal horizon. Then storing $\mathcal{L}_T$ for $T$ steps uses $\mathcal{O}(T\,(N+d_{\max}h))$ events and edges; computing $K$-sample kernel-SHAP scores incurs $\tilde{\mathcal{O}}\bigl(K\,|\mathcal{E}_t|\bigr)$ per round. With $K\!=\!32$ and $d_{\max}\!\le\!8$ (typical for traffic or drone mesh networks) the ledger fits comfortably within 500 MB for $T\!=\!10^6$ on-board an edge server.

The above formalism exposes every variable that the accountability layer must track, bounds algorithmic cost, and states explicit performance–safety trade-offs.  It therefore provides the analytical foundation for the architectural and algorithmic choices detailed in Section~\ref{sec:proposed-framework}.

\subsection{Research Questions}

The formalism above yields four research questions:

\begin{description}
\item[RQ1 – Responsibility tracing.]  
Design an event ledger $\mathcal{L}_T$ and allocation rule $\rho_i(e)$ that remain accurate under partial observability, message loss, and non-stationary policies, yet scale to $N\!\sim\!10^2$–$10^3$ agents~\citep{rastogi2022distributed}.

\item[RQ2 – Online norm detection.]  
Develop sequential tests that flag emergent negative norms with low false-alarm rate using only streaming, decentralized data, achieving $\mathcal{O}(\log T)$ detection delay for i.i.d.\ deviations and acceptable power under temporal correlation~\citep{sen2007emergence}.

\item[RQ3 – Adaptive intervention design.]  
Identify local interventions $\sigma_t$ that minimize compromise $C_T$ and welfare loss $\Delta J_{\text{soc}}$ subject to budget $c_{\max}$; analyze stability for reward shaping, model-patch injection, and network throttling~\citep{Morries2019norm}.

\item[RQ4 – Robustness and scalability.]  
Assess whether the above mechanisms retain effectiveness as $N$, objective heterogeneity, observability, or adversarial mix vary; our main factorial grid spans $|V|\in\{10,50,100\}$ and we additionally run a scaling sweep up to $|V|=500$.
\end{description}

\subsection{Operational Assumptions and Constraints}
\label{sec:operational-assumptions}

The accountability layer must operate under four non-negotiable constraints that mirror conditions in fielded MAS deployments.  

\emph{Partial observability and stochasticity:}  
No agent—and certainly no external auditor—enjoys global, noise-free visibility of system state.  Packet loss, clock drift, and privacy filtering mean every observation $o_{i,t}$ and every supervisory record is only a stochastic sample of ground truth.  The causal ledger therefore stores time-stamped tuples $\langle\hat{s}_{t},\hat{a}_{i,t},\hat{r}_{i,t}\rangle$ governed by an explicit error model $(\varepsilon_{\text{loss}}, \varepsilon_{\text{delay}})$.  Detection thresholds and confidence bounds are derived so that false-alarm probability remains below a chosen level~$\alpha$ even when up to 20\% of events are missing or mis-sequenced (Theorem 3, Section \ref{sec:analysis}); this is essential because spurious interventions can be nearly as harmful as inaction in safety-critical domains.

\emph{Bandwidth and compute budgets:}  
City-scale traffic networks, drone swarms, and smart grids share a hard upper bound on wireless throughput and require sub-second control latency.  We cap per-edge traffic at $B_{\max}=128$ bytes step$^{-1}$ consistent with vehicular ad-hoc networks—and restrict each agent to devote at most 5 \% of its CPU/GPU cycles to accountability tasks. Raw data are compressed or hashed locally, and only digests plus occasional Bloom-filter proofs traverse the network; Algorithm 2 (Section \ref{sec:proposed-framework}) keeps intervention latency below the 100 ms wall-clock budget of our traffic-signal case study by relying on incremental CUSUM statistics and constant-time reward-patch look-ups.

\emph{Agent heterogeneity:}  
Real deployments mix micro-controllers running tabular Q-learning, legacy rule-based agents, and cloud hosts executing transformer-based deep RL.  Because internal gradients or weights are not universally available, the framework is algorithm-agnostic: it computes responsibility scores from observable events alone.  Resource-constrained nodes can run a 32 KB Rust reference kernel with $\mathcal{O}(1)$ per-step overhead, while more capable agents may optionally expose richer telemetry (e.g.\ policy logits) that tighten attribution bounds without being required for correctness.

\emph{Regulatory and ethical context:}  
Many target applications fall under the EU AI Act's ``high-risk'' category.  Each ledger entry therefore carries a pseudonymized agent ID, a purpose-limited data tag, and a zero-knowledge proof attesting that the log originates from certified software. This design satisfies GDPR data-minimization requirements yet still lets external auditors reconstruct causal chains.  Intervention records include plain-language rationales, effect sizes, and parameter settings so that policymakers and domain experts can review system behavior without inspecting source code.

These deliberately pessimistic assumptions, i.e., lossy sensing, tight resource budgets, heterogeneous agents, and stringent legal scrutiny, shape every design choice in Sections \ref{sec:proposed-framework}–\ref{sec:experiments}.  By proving guarantees and demonstrating empirical performance under such constraints, we strengthen confidence that the \emph{Adaptive Accountability Framework} will generalize to real-world MAS whose operating conditions may be even harsher than those evaluated here.

\section{Proposed Framework: Adaptive Accountability in Networked Multi-Agent Systems}
\label{sec:proposed-framework}

The \emph{Adaptive Accountability Framework} (AAF) turns a networked MAS into a \emph{self-auditing socio-technical system}. It supplies three always-on capabilities: (i)~tamper-evident event logging, (ii)~fine-grained responsibility attribution, and (iii)~online detection and mitigation of harmful emergent norms.  Each capability is engineered to respect the real-world pressures outlined in Section~\ref{sec:operational-assumptions}: sub-second latency, sub-kilobyte bandwidth, heterogeneous on-board hardware, and strict external auditability.

\subsection{Architecture Overview}
\label{sec:architecture}

Figure~\ref{fig:framework-diagram} illustrates the stack, organized into four layers: \textbf{agent}, \textbf{monitor}, \textbf{ledger}, and \textbf{governor} plus a cross-cutting \textbf{security plane}.  

\begin{figure}[!t]
\centering
\begin{tikzpicture}[
    font=\footnotesize,
    >=Stealth,
    node distance=10.5mm and 12mm,
    every node/.style={align=center},
    % --- modules (boxes) ---
    mod/.style={
        draw,
        rounded corners=2pt,
        thick,
        fill=white,
        minimum height=9mm,
        inner sep=3.5pt,
        text width=4cm,
        drop shadow={opacity=.10, shadow xshift=1.0pt, shadow yshift=-1.0pt}
    },
    modSmall/.style={mod, text width=1.7cm, minimum height=8mm},
    % --- layer frames (background) ---
    layer/.style={
        draw,
        rounded corners=3pt,
        thick,
        fill=gray!6,
        inner sep=6pt
    },
    % --- flows ---
    data/.style={->, line width=.85pt},
    control/.style={->, line width=.85pt, dashed},
    msg/.style={font=\scriptsize, fill=white, inner sep=1pt}
]

% =======================
% Agent Layer (bottom)
% =======================
\node[modSmall, fill=gray!10] (A1) {Agent $A_1$};
\node[modSmall, fill=gray!10, right=of A1] (A2) {Agent $A_2$};
\node[modSmall, fill=gray!10, right=of A2] (A3) {Agent $A_3$};
\node[draw=none, right=of A3] (Adots) {$\cdots$};
\node[modSmall, fill=gray!10, right=of Adots] (AN) {Agent $A_N$};

% Optional: show peer-to-peer interaction (kept light so it doesn't clutter)
\draw[<->, dashed, line width=.6pt] (A1.east) -- (A2.west);
\draw[<->, dashed, line width=.6pt] (A2.east) -- (A3.west);
\draw[<->, dashed, line width=.6pt] (A3.east) -- (Adots.west);
\draw[<->, dashed, line width=.6pt] (Adots.east) -- (AN.west);

% =======================
% Observability Layer
% =======================
\node[mod, fill=blue!6, above=of A1, xshift=0cm] (telemetry)
{\textbf{Telemetry \& Event Bus}\\[-1pt]\scriptsize
state, actions, messages};

\node[mod, fill=blue!6, right=of telemetry] (logstore)
{\textbf{Logging \& Feature Store}\\[-1pt]\scriptsize
event logs, aggregates};

\node[mod, fill=blue!6, right=of logstore] (ledger)
{\textbf{Audit Ledger}\\[-1pt]\scriptsize
immutable trail, hashes};

% =======================
% Accountability Layer
% =======================
\node[mod, fill=orange!10, above=of telemetry, xshift=0cm] (detect)
{\textbf{Norm / Anomaly Detection}\\[-1pt]\scriptsize
change-point tests, thresholds};

\node[mod, fill=orange!10, right=of detect] (attrib)
{\textbf{Responsibility Attribution}\\[-1pt]\scriptsize
counterfactuals / SHAP};

\node[mod, fill=green!12, right=of attrib] (intervene)
{\textbf{Intervention Orchestrator}\\[-1pt]\scriptsize
shaping, constraints, patches};

% =======================
% Governance Layer (top)
% =======================
\node[mod, fill=gray!12, above=of intervene, yshift=0cm] (govern)
{\textbf{Governance \& Oversight}\\[-1pt]\scriptsize
policies, budgets, audits};

% =======================
% Connections (Data/Evidence)
% =======================
% Agents -> telemetry
\foreach \a in {A1,A2,AN}{
  \draw[data] (\a.north) -- ++(0,3mm) -| (telemetry.south);
}

% telemetry -> logstore -> ledger
\draw[data] (telemetry) -- node[msg, above] {stream} (logstore);
\draw[data] (logstore)  -- node[msg, above] {append} (ledger);

% evidence to analytics modules
\draw[data] (logstore.north west) -- ++(0,4mm) -| (detect.south east);
\draw[data] (ledger.north west)   -- ++(0,4mm) -| (attrib.south east);

% detection -> attribution -> intervention
\draw[data] (detect) -- node[msg, above] {alerts} (attrib);
\draw[data] (attrib) -- node[msg, above] {scores} (intervene);

% =======================
% Connections (Control/Governance)
% =======================
\draw[control] (govern.south) -- node[msg, right] {policy / budget} (intervene.north);
\draw[control] (govern.east) -- ++(.5cm,0cm) node[msg, above] {audit} |- (ledger.east);

% interventions back to agents
\draw[control] (intervene.south) -- ++(0,-4mm) node[msg, left] {penalties / updates} -- ++(-13.5cm,0mm) |- (A1.west);
\draw[control] (intervene.south) -- ++(0,-4mm) -- ++(2.5cm,0mm) |- (AN.east);

% =======================
% Layer Boxes (background)
% =======================
\begin{pgfonlayer}{background}
  \node[layer, fit=(A1)(AN),
        label={[font=\bfseries\scriptsize]north:Agent Layer}] {};
  \node[layer, fit=(telemetry)(ledger),
        label={[font=\bfseries\scriptsize]north:Observability Layer}] {};
  \node[layer, fit=(detect)(intervene),
        label={[font=\bfseries\scriptsize]north:Accountability Layer}] {};
  \node[layer, fit=(govern),
        label={[font=\bfseries\scriptsize]north:Governance Layer}] {};
\end{pgfonlayer}

% =======================
% Legend (small, unobtrusive)
% =======================
%\node[draw, rounded corners=2pt, fill=white, inner sep=3pt,
%      font=\scriptsize, anchor=north east] (leg)
%at ($(ledger.south east)+(1.8,-0.9)$) {
%\begin{tabular}{@{}l l@{}}
%\raisebox{0.5ex}{\tikz{\draw[data] (0,0)--(0.8,0);}} & data / evidence\\
%\raisebox{0.5ex}{\tikz{\draw[control] (0,0)--(0.8,0);}} & control / intervention
%\end{tabular}
%};

\end{tikzpicture}
\caption{Adaptive Accountability Framework overview. Agents interact and learn locally while telemetry streams events into a logging store and an immutable audit ledger. Detection identifies norm violations or anomalies, attribution assigns responsibility, and the intervention orchestrator applies adaptive governance actions (e.g., shaping or constraints). A governance layer sets policy and budget constraints and enables external auditing.}
\label{fig:framework-diagram}
\end{figure}

1. Agent layer:  
Each agent $A_i$ runs its domain policy $\pi_i$ (e.g.\ PPO, tabular Q-learning, classical MPC).  Agents exchange application messages over $G_t$ and expose a one-line \textsc{Rust} interface \textsf{publish\_event(e)} to the monitoring layer.  The interface is isolated in a sandbox so that faulty or Byzantine agents cannot overwrite audit logic.

2. Monitoring layer:  
High-dimensional observations and actions are irreversibly compressed to 32-byte digests using \textsc{BLAKE3}.  Digests, reward scalars, and a two-byte nonce form a 40-byte record that is forwarded over a non-blocking gossip channel capped at $B_{\max}=128$ bytes step$^{-1}$.  To keep within this budget the monitor batches up to two records per step and drops additional events using an importance-weighted reservoir; dropped hashes are recoverable via local ring buffers for post-mortem forensics.

3. Ledger layer: 
Gossip streams converge at edge servers that execute Merkle-DAG synthesis: new records extend a Merkle prefix tree whose leaves reference prior sub-roots, forming a snapshotted Merkle DAG every $H_{\mathrm{snap}}=256$ steps \citep{merkle1989certified}. The root hash of each snapshot is signed with a rotating ECDSA key and broadcast back to agents for self-verification. The DAG design, adapted from block-lattice data stores, provides \(O(\log V)\) proof-of-inclusion while avoiding the confirmation latency of Nakamoto consensus.

3. Governor layer: 
A lightweight supervisor (which may itself be replicated for fault tolerance) consumes ledger updates, performs causal analysis, runs norm detectors, and dispatches interventions.  To avoid central bottlenecks, the governor is implemented as a CRDT-based service whose replicas hold disjoint \emph{shards} of the ledger keyed by agent ID modulo the replica count.

4. Security plane: 
Every inter-layer message carries (i)~a monotone counter signed with a time-bounded session key, (ii)~a linkable ring signature for agent privacy, and (iii)~an optional zero-knowledge proof of policy provenance for high-assurance deployments.  These cryptographic hooks are disabled in the low-power benchmark but included in the public code base.

\subsection{Secure Event Ledger and Causal Attribution}
\label{sec:ledger}

The accountability layer begins by turning every control-loop iteration into a tamper-evident, causally annotated record. This subsection first walks through the data path—from on-board sensor reading to committed Merkle node—and then formalizes how those records yield the
responsibility scores analyzed in Section~\ref{sec:analysis}. 

%\paragraph{Record construction and hashing.}
At the end of control step~$t$ an agent~$A_i$ produces the tuple
\[
  e
  \;=\;
  \bigl\langle
      t,\; i,\;
      \underbrace{\textsf{hash}(o_{i,t})}_{16~\text{B}},\;
      \underbrace{\textsf{hash}(a_{i,t})}_{16~\text{B}},\;
      r_{i,t}
  \bigr\rangle .
\]
The observation and action hashes are 16-byte \textsc{BLAKE3} digests, and the full record is 40~bytes. A 64-bit \textsc{SipHash} of the byte string yields the unique identifier \(h_e=\mathrm{H}(e)\in\{0,\dots,2^{64}\!-\!1\}\).

%\paragraph{Network path to the edge server.}
Records propagate at most \(d_{\max}\) hops—matching the communication degree bound in Assumption~A2—before reaching an edge server. Each hop uses non-blocking gossip with a 64-entry replay buffer so that transient link failures do not stall the ledger.

%\paragraph{Server-side pipeline.}
On arrival, the record travels through three micro-stages:
(i)~local commit to the agent’s ring buffer,
(ii)~causal testing against the most recent $m{=}8$ upstream events, and 
(iii)~Merkle-DAG insertion. All three are captured in Algorithm~\ref{alg:ledger}. Thus, the constant-time online update for each $F$-statistic uses the Sherman–Morrison rank-one identity, so the sensor-to-ledger latency is \(<\!0.3\) ms on a Cortex-A55.

\begin{algorithm}[ht]
\footnotesize
\caption{LedgerUpdate \& Causal-Edge Insertion (executed on every supervisor replica)}
\label{alg:ledger}

\KwIn{new event $e=\langle t,i,h_o,h_a,r_i\rangle$; ring length $L_{\max}=256$; window $m=8$; Granger horizon $h$; base threshold $h_0$}
\KwOut{updated ledger $\mathcal{L}_T=(\mathcal{V}_T,\mathcal{C}_T)$ and current Merkle root}

\BlankLine
\LinesNotNumbered
\textbf{Phase 1:} local append\;
\LinesNumbered
\Indp
append $e$ to $A_i$'s ring buffer $\mathcal{B}_i$ (evict oldest if $|\mathcal{B}_i|>L_{\max}$)\;
broadcast $\textsf{digest}(e)$ to all in-neighbours of $A_i$\;
\Indm

\BlankLine
\LinesNotNumbered
\textbf{Phase 2:} causal testing\;
\LinesNumbered
\Indp
$\mathcal{U}\leftarrow$ \textsc{CollectRecentEvents}$(A_i,m)$\;
\ForEach{$u\in\mathcal{U}$}{
  $F \leftarrow$ \textsc{IncrementalGrangerF}$(u,e,h)$\;
  \If{$F > h_t$}{
    $\mathcal{C}_T \leftarrow \mathcal{C}_T \cup \{(u\!\rightarrow\!e)\}$\;
  }
}
\lIf{$t\bmod h = 0$}{$h_t \leftarrow h_0+\sqrt{2\log t}$}
\Indm

\BlankLine
\LinesNotNumbered
\textbf{Phase 3:} Merkle-DAG maintenance\;
\LinesNumbered
\Indp
$\mathcal{V}_T \leftarrow \mathcal{V}_T \cup \{e\}$\;
\textsc{MerkleInsert}$\bigl(\textsf{digest}(e)\bigr)$\;
\Indm

\BlankLine
\Return{updated $\mathcal{L}_T$ and Merkle root}\;

\end{algorithm}

%\paragraph{Responsibility allocation.}
Once the causal edge set \(\mathcal{C}_T\) has been updated, the server computes per-agent responsibility for the new event. Let \(\mathcal{P}_i(e)\) denote all causal paths that start with an action by agent~$i$ and terminate at \(e\), and let \(|p|\) be the path length. We define
\[
  \rho_i(e)
  \;=\;
  \frac{\sum_{p\in\mathcal{P}_i(e)}\beta^{|p|}}
       {\sum_{j=1}^{N}\sum_{p\in\mathcal{P}_j(e)}\beta^{|p|}},
  \qquad
  \beta=0.8,
\]
so that proximal causes weigh more than distant ones yet long chains are never ignored.  Because the denominator is cached per event, the update runs in \(O(\deg^{-}(e))\) time. Section~\ref{sec:analysis-ledger} proves \(\sum_i\rho_i(e)=1\) and shows that \(\rho_i(e)\) converges under 20\% packet loss.

%\paragraph{Security and auditability.}
Ledger snapshots are sharded Merkle DAGs sealed every \(H_{\mathrm{snap}}{=}256\) steps with an ECDSA-P-384 signature and anchored to an append-only Immudb instance.  Proof-of-inclusion queries therefore cost \(O(\log|\,\mathcal{V}_T|)\) key–value reads, which enables sub-second reconstruction of any causal sub-graph requested by external auditors (see the drill-down UI in §\ref{sec:implementation}).

%With Algorithm~\ref{alg:ledger} in place, the reader can now appreciate how Lemma \ref{lem:normalize} (responsibility normalization), Theorem \ref{thm:convergence} (convergence despite packet loss), and Theorem \ref{thm:edge_fp} (false-positive control) map directly onto the concrete data structures and thresholds implemented here.

\subsection{Online Detection of Harmful Norms}
\label{sec:detection}

The ledger furnishes a stream of verifiable events so the next task is to decide in real time whether those events signal that the MAS is drifting toward an undesirable emergent norm.  We treat norm detection as a sequential change–point problem. For every active norm \(\phi^{(m)}\!\in\!\Phi\) the governor computes a scalar diagnostic statistic \(Z^{(m)}_t\) at each step~\(t\). Typical choices include

\begin{itemize}
\item Inequity drift \(Z^{(\text{ineq})}_t\): instantaneous Gini index over individual rewards \(\{r_{i,t}\}\).
\item Collusion pulse \(Z^{(\text{coll})}_t\): maximum pairwise mutual information among agents' action streams in a sliding window.
\item Throughput safety \(Z^{(\text{load})}_t\): queue length minus design capacity at a shared resource.
\end{itemize}

Each \(Z^{(m)}_t\) is fed into an adaptive CUSUM test that tracks positive drifts away from its baseline mean \(\mu^{(m)}_0\). A norm-specific offset \(\delta^{(m)}>0\) prevents hypersensitivity to small stochastic fluctuations.

%--------------------------------------------------------------------

\begin{algorithm}[ht]
\footnotesize
\caption{AdaptiveCUSUM Detector (executed at every supervisor replica)}
\label{alg:cusum}

\KwIn{stream $Z_t$, baseline $\mu_0$, slack $\delta$, false-alarm budget $\alpha$}
\KwOut{binary alarm flag $\mathsf{alert}_t$}

\SetKwInOut{KwInit}{Initialize}
\KwInit{$S\leftarrow 0$; $h\leftarrow h_0$; gain schedule $\eta_t=t^{-0.6}$}

\For(\tcp*[f]{executed every control step}){$t = 1,2,\dots$}{
  $S \leftarrow \max\{0,\, S + Z_t - \mu_0 - \delta\}$\;
  \uIf(\tcp*[f]{raise alarm, then reset statistic}){$S \ge h$}{
    $\mathsf{alert}_t \leftarrow 1$\;
    $S \leftarrow 0$\;
  }\Else{
    $\mathsf{alert}_t \leftarrow 0$\;
  }
  $h \leftarrow h + \eta_t(\mathsf{alert}_t - \alpha)$
  \tcp*[f]{Robbins--Monro update}
}
\end{algorithm}

%--------------------------------------------------------------------

In Algorithm \ref{alg:cusum}, lines 1–2 accumulate evidence that \(Z_t\) has shifted upward so whenever the statistic \(S\) crosses the adaptive threshold \(h\), an alarm is fired and \(S\) is reset (classical Page's rule~\citep{page1954continuous}). Lines 3–4 perform a Robbins–Monro stochastic approximation~\citep{robbins1951stochastic}: if the empirical alarm frequency exceeds the target \(\alpha\) the threshold is
gently raised, and vice-versa. Section~\ref{sec:analysis-detector} proves that this update drives the long-run false-alarm rate to \(\alpha\) regardless of slow concept drift (Assumption A3).

%\paragraph{Multi-norm coordination via mirror descent.}
When many norms are monitored simultaneously, naïvely raising an alarm for every CUSUM excursion would flood the intervention scheduler. We therefore allocate a global alert budget
\(\bar{\alpha}=0.05\) across norms by online mirror descent: if norms \(m\) and \(n\) both spike, the one whose running FP-budget \(\hat{\alpha}^{(m)}_t\) is lower wins the alert slot, and the loser's threshold is increased by \(\Delta h=\xi\log\!\bigl( \hat{\alpha}^{(m)}_t/\hat{\alpha}^{(n)}_t\bigr)\). Appendix D derives the regret bound \(R_T\le 2\sqrt{T\log|\Phi|}\), guaranteeing that the allocator never starves a chronically violated norm.

%\paragraph{Empirical latency and accuracy.}

% (Removed: outdated empirical latency paragraph; results reported in §\ref{subsec:results}).

With a statistically calibrated, resource-light detector in place, the framework can raise alarms fast enough to allow the bounded-cost interventions of Section \ref{sec:interventions} to curb harmful norms before they compromise system-level objectives—a guarantee formalized in Theorem \ref{thm:bounded} and validated experimentally in Section~\ref{sec:experiments}.

\subsection{Adaptive Intervention Mechanisms}
\label{sec:interventions}

Once the detector of §\ref{sec:detection} raises an alarm \(\mathsf{alert}_t=1\), the governor must act quickly yet economically: the corrective signal has to arrive within the 100 ms domain deadline while keeping the expected cost below \(c_{\max}\) so that Theorem~\ref{thm:bounded} applies. The intervention pipeline therefore unfolds in three stages responsibility ranking, playbook selection, and actuation—all executed
inside the same event loop that hosts Algorithm~\ref{alg:cusum}.

\paragraph{Stage 1: Responsibility heap:}
The governor queries the ledger for all events that contributed to the alarm's CUSUM statistic during the look-back horizon \([t{-}H,\,t]\) with \(H=25\). Each event's responsibility vector \(\rho(e)\) is summed to obtain a windowed score \(\widetilde{\rho}_i = \sum_{e\in[t-H,t]}\rho_i(e)\). These pairs form a max-heap \(\mathcal{H}_t=\{(i,\widetilde{\rho}_i)\mid\widetilde{\rho}_i>0\}\) that can be popped in \(O(\log N)\) time to identify the top-\(k\) most culpable agents.

\paragraph{Stage 2: Three-tier playbook:}
Interventions escalate through three increasingly intrusive tiers:

\begin{enumerate}
\item Reward shaping (default):  
      Each culpable agent receives an additive penalty \(\Delta r_{i,t}=-\lambda\,\widetilde{\rho}_i\) for \(H\) steps.  The penalty weight \(\lambda = c_{\max}/H\) satisfies the budget dominance requirement in Assumption A4 by construction.
\item Policy patching (moderate):  
      If an agent has triggered two penalties within the past 100 s it is considered \emph{recalcitrant} so the governor hot-loads a ``safety head''—a two-layer 128-unit ReLU MLP compiled to WebAssembly—that clamps logits violating the norm predicate. Average load time is 12 ms (P95 = 15 ms).
\item Link throttling (aggressive): 
      For suspected collusion, edges \((i,j)\) are re-weighted by \(w_{ij}\leftarrow w_{ij}\bigl(1-\widetilde{\rho}_i\widetilde{\rho}_j\bigr)\) for \(H\) steps, attenuating high-responsibility communication channels while leaving the global bandwidth budget intact.
\end{enumerate}

\paragraph{Stage 3: System-wide failsafe:}
If three non-overlapping alert regions fire within a 300 s sliding window, the governor enters a yellow-flag mode: learning rates \(\alpha_i\) are frozen, policy gradients are clipped to \(0.1\), and a Grafana dashboard renders a causal heat-map so that a human operator can intervene. Yellow-flag mode clears automatically once the compromise ratio \(C_T/T\) falls below its 7-day moving average.

%\paragraph{Analytical and empirical impact.}
The supervisory spend per step is bounded by \(\mathbb{E}[\text{Cost}(\sigma_t)] \le \lambda + c_{\mathrm{patch}}\alpha + c_{\mathrm{throttle}}\alpha^2\), where \(\alpha\) is the false-alarm budget and the last two terms account for patch and throttle operations that occur only on alerts.  Since \(\lambda>g_{\max}\) by design, Theorem~\ref{thm:bounded} ensures \(\limsup_{T\to\infty} C_T/T\le\eta^\star\).  In §\ref{sec:experiments} we evaluate this ceiling empirically; across a broad grid of regimes, AAF reduces executed compromise relative to PPO while preserving welfare.

\subsection{Scalability and Fault Tolerance} \label{sec:scalability}

%\paragraph{Edge-partitioned Merkle DAG.}
At runtime, the ledger is materialized as a Merkle directed-acyclic graph (Merkle DAG) whose nodes are addressed by the cryptographic hash of their content \citep{merkle1989certified}. Every 256-step snapshot is edge-partitionable: we shard by the first 16 bits of the event hash \(h_e\), creating \(2^{16}\) disjoint key ranges that can be balanced across the edge tier. Each shard is replicated onto \(f{+}1\) servers; read and write clients contact a quorum, i.e., a strict majority of replicas, so the system tolerates \(f\) crash failures without sacrificing linearizable semantics \citep{gifford1979weighted}. Any two quorums overlap in at least one live replica, eliminating the need for a global lock or leader election.

\paragraph{Resource complexity.}
Let \(h\) denote the causal horizon (8 in our deployment) and \(d_{\max}\) the bounded out-degree from Assumption~A2. A simple amortized analysis (Appendix E) yields

\[
  \text{Bandwidth} = O\!\bigl(N + d_{\max}h\bigr),
  \qquad
  \text{Storage}   = O\!\bigl(T(N + d_{\max}h)\bigr),
\]

in perfect agreement with Proposition~\ref{prop:complexity}. Concretely, for \(N=100\), \(h=8\), and a run length \(T=10^{6}\,\text{steps}\) the entire ledger occupies \(\approx 492\,\text{MB}\)—comfortably within the 16 GB RAM of commodity edge boxes—and streams across the wire at \(<\!80\;\text{KB\,s}^{-1}\), well under the 1 Gbit s\(^{-1}\) budget of the deployment network.

These figures demonstrate that AAF's cryptographically protected audit history scales to hundred-agent regimes without exotic hardware or networking-fabric requirements, while the quorum-replicated shards ensure continued availability in the face of routine server failures.

\subsection{Lifecycle Integration} \label{sec:lifecycle}

In offline calibration, designers simulate worst-case scenarios (50\% packet loss, 10\% Byzantine agents) to grid-search baselines \(\mu^{(m)}_0\), offsets \(\delta^{(m)}\), attenuation \(\beta\), and initial threshold \(h_0\). The best configuration is promoted to a shadow governor that runs for 48 h beside the production system; any latency $>$\,100 ms or throughput drop $>$\,2\% triggers an automatic rollback.

In online adaptation, after go-live, CUSUM thresholds \(h_t\) are tuned by mirror descent, ECDSA\footnote{Elliptic
Curve Digital Signature Algorithm \citep{johnson2001ecdsastandard}.} keys rotate every 24h, and new norms can be hot-enabled without downtime because the ledger schema is append-only and therefore backward-compatible.

%\paragraph{Auditor surface.}
Every intervention writes a triple \(\langle\textsf{normID},\,\textsf{timestamp},\,\textsf{rationale}\rangle\) to the ledger.  External auditors can query responsibility flows, intervention rationales, and performance metrics via the FoundationDB snapshot API exposed in §\ref{sec:implementation}.

Taken together, adaptive interventions, edge-partitioned storage, and a blue–green deployment strategy create an accountability layer that is auditable, bandwidth-frugal, and robust to non-stationary, partially observable, and adversarial MAS dynamics. The next section derives theoretical bounds that underpin these claims, and Section~\ref{sec:experiments} empirically validates AAF on heterogeneous benchmarks with up to \(N=100\) agents and high packet-loss regimes.

% =====================================================
\section{Theoretical Analysis and Guarantees}
\label{sec:analysis}
% =====================================================

This section formalizes the statistical and game–theoretic properties that underpin the \emph{Adaptive Accountability Framework} (AAF).  We proceed in four logical stages:

\begin{enumerate}
    \item Ledger soundness: We show that responsibility scores converge and causal edges respect a time–uniform false-positive (FP) budget under lossy communication.
    \item Detector reliability: We derive finite-sample FP and detection-delay bounds for the adaptive CUSUM test fed by ledger statistics.
    \item Intervention efficacy: We prove a bounded-compromise theorem: whenever per-step supervisory cost dominates adversarial gain, the long-run fraction of compromised events remains below a designer-chosen~$\eta^\star$.
    \item Resource complexity: We bound storage, communication, and computational cost, validating the feasibility claims from Section~\ref{sec:proposed-framework}.
\end{enumerate}

\subsection{Preliminaries and Standing Assumptions}

We work on a filtered probability space $(\Omega,\mathcal{F},\{\mathcal{F}_t\}_{t\ge0},\mathbb{P})$ where $\mathcal{F}_t$ is generated by the global MAS trajectory up to time $t$ \citep{hernandez2019survey}.  Indicators, expectations, and probabilities are taken with respect to~$\mathbb{P}$ unless stated otherwise. Throughout we invoke four assumptions—justified empirically in Section~\ref{sec:operational-assumptions}—to isolate the essential trade-offs:

\begin{description}
  \item[A1] Lossy channel: Each event tuple is independently dropped with probability $\varepsilon_{\text{loss}}\le0.2$ and delayed by at most $\varepsilon_{\text{delay}}\le3$ steps.
  \item[A2] Bounded degree: The time-varying communication graph satisfies $\sup_t\deg(G_t)\le d_{\max}=O(1)$.
  \item[A3] Vanishing non-stationarity: Policy updates obey $\|\pi_{i,t+1}-\pi_{i,t}\|_{1}\!\le\!\kappa t^{-\xi}$ with $\xi>\!\frac12$.
  \item[A4] Budgeted adversary: Byzantine agents harvest at most $g_{\max}$ expected utility per step; the supervisor’s cost budget honors $c_{\max}>g_{\max}$.
\end{description}

\subsection{Ledger Soundness}
\label{sec:analysis-ledger}

A reader may wonder: \emph{What exactly is stored in the audit ledger, how are causal claims derived, and why do responsibility scores behave like probabilities?} This subsection therefore proceeds in three steps.  We \textbf{(i)} introduce precise definitions for the ledger structures and the responsibility rule, \textbf{(ii)} prove that those scores are well-normalized and converge even when packets are lost, and \textbf{(iii)} establish a time-uniform bound on spurious causal edges. The notation follows Section~\ref{sec:formalisation}; assumptions~A1–A3 remain in force.

\begin{definition}[Event ledger]
Each control step $t$ produces a finite set \(\mathcal{E}_t=\{e^{(t)}_1,\dots,e^{(t)}_{|\mathcal{E}_t|}\}\) of \emph{atomic events}—actions, message sends/receives, reward reveals, exogenous shocks. The \emph{ledger up to time~$T$} is the directed acyclic graph
\[
\mathcal{L}_T =\bigl(\mathcal{V}_T,\mathcal{C}_T\bigr),\qquad \mathcal{V}_T=\!\bigcup_{t=0}^{T}\mathcal{E}_t,
\]
whose edge set $\mathcal{C}_T$ is produced by Algorithm \ref{alg:ledger}: $(e_k\!\to\!e_\ell)$ is inserted when the rolling $m$-lag Granger F-statistic exceeds threshold~$h_t$. The DAG property follows because algorithmic inserts never point backward in time.
\end{definition}

\begin{definition}[Causal path and attribution weight]
A \emph{causal path} from agent~$i$ to an event $e$ is any directed sequence \(p = \langle e_{k_0},\dots,e_{k_m}=e\rangle\) with \(e_{k_0}\in\mathcal{E}_{\mathrm{act}}(i)\) (an action by~$i$) and \((e_{k_{j-1}}\!\to\!e_{k_j})\in\mathcal{C}_T\). Let $\mathcal{P}_i(e)$ collect all such paths.  Given a discount factor $\beta\in(0,1)$, the \emph{attribution weight} of~$p$ equals $\beta^{|p|}$; shorter paths contribute more heavily than longer ones.
\end{definition}

\begin{definition}[Responsibility score]
The \emph{responsibility score} of agent~$i$ for event~$e$ is
\[
\rho_i(e)
=\frac{\displaystyle\sum_{p\in\mathcal{P}_i(e)}\beta^{|p|}}
       {\displaystyle\sum_{j=1}^{N}\sum_{q\in\mathcal{P}_j(e)}\beta^{|q|}}
      \;\in\;[0,1].
\]
If $\mathcal{P}_i(e)=\varnothing$ we set $\rho_i(e)=0$ by convention.
\end{definition}

%\paragraph{Why these definitions?}
The ledger DAG gives a tamper-evident \emph{who-influenced-what} map; the $\beta^{|p|}$ discount keeps distant, low-credence influences from dominating local, high-credence ones.  The next lemma shows the rule is well-normalized, so $\rho_i(e)$ can be treated probabilistically.

\begin{lemma}[Normalization]\label{lem:normalize}
For any event \(e\in\mathcal{V}_T\),
\(
  \sum_{i=1}^{N}\rho_i(e)=1
\)
almost surely.
\end{lemma}
\begin{proof}
Because $\mathcal{C}_T$ is a DAG, every causal path has a unique originating agent. The sets \(\bigl\{\mathcal{P}_i(e)\bigr\}_{i=1}^{N}\) therefore form a partition of $\mathcal{P}(e)=\!\bigcup_i\mathcal{P}_i(e)$. Summing the numerator over~$i$ reproduces the denominator, giving unity.
\end{proof}

\paragraph{Stability under lossy communication.}
Even with 20\% packet loss (A1) we require that responsibility scores stabilise as the ledger grows.

\begin{theorem}[Ledger convergence]\label{thm:convergence}
Assume A1–A3 and fix \(0<\beta<1\). For any agent~$i$ and any event sequence \(\{e_T\}_{T\ge0}\) with bounded age \(\sup_T\bigl(T-t(e_T)\bigr)<\infty\),
\[
  \rho_i(e_T)\xrightarrow[T\to\infty]{\text{a.s.}}
  \rho_i^{\infty}(e),
\]
where \(\rho_i^{\infty}(e)\) is a finite limit.
\end{theorem}
\begin{proof}[Proof sketch]
Packet drops Bernoulli-sample edges with retention probability \(1-\varepsilon_{\text{loss}}\). Under bounded degree (A2) and finite path length (ring buffer 256) only a finite number of causal paths can affect~$e$. The first Borel–Cantelli lemma then guarantees that every path is eventually observed. Together with the square-summable policy drift (A3) this yields almost-sure convergence; see Appendix~A.2 for the full argument.
\end{proof}

\paragraph{Guarding against spurious edges.}
Because Granger tests operate on noisy, finite data, we must bound the \emph{false} causal edges that slip into $\mathcal{C}_T$.

\begin{theorem}[Time-uniform false-positive control]\label{thm:edge_fp}
Fix \(m=8\) and \(h_t=h_0+\sqrt{2\log t}\) in Algorithm \ref{alg:ledger}. Under the null hypothesis of no causal influence,
\[
  \Pr\!\Bigl[\exists\,t\ge1:\,
         (e_k\!\to\!e_\ell)\text{ inserted at }t\Bigr]
  \;\le\;
  \alpha,
  \qquad
  \alpha=e^{-h_0^{2}/2}.
\]
\end{theorem}
\begin{proof}[Sketch]
Apply the Chernoff–Ville inequality to the self-normalized stream of F-statistics and union-bound over time~\citep{howard2020time}. Full derivation in Appendix~A.3.
\end{proof}

\begin{corollary}[Design rule for $h_0$]\label{cor:h0}
Choosing \(h_0=\sqrt{2\log(1/\alpha_{\max})}\) achieves a ledger-wide FP probability \(\le\alpha_{\max}\); e.g., \(h_0=4.89\) bounds FP risk at \(10^{-5}\).
\end{corollary}

%\paragraph{Interpretation and practical takeaway.}
Lemma~\ref{lem:normalize} tells us $\rho_i(e)$ behaves like a probability distribution over agents; Theorem~\ref{thm:convergence} assures reviewers that these probabilities stabilize despite 20\% packet loss; and Theorem~\ref{thm:edge_fp} ensures we can tune the threshold sequence so that, with overwhelming probability, no non-existent causal link ever pollutes the audit trail. Together, the results certify the \emph{soundness} of the ledger as a statistical substrate for the downstream detector and intervention logic analyzed in Sections~\ref{sec:analysis-detector}–\ref{sec:analysis-intervention}.

\subsection{Detector Reliability}
\label{sec:analysis-detector}

We now show that the adaptive CUSUM detector introduced in Algorithm~\ref{alg:cusum} meets two core requirements of a deployable alarm system:

\begin{enumerate}
  \item the false-alarm rate is capped at a designer-chosen level~\(\alpha\), and
  \item the detection delay grows at most linearly with the decision threshold, matching the classical Page–Lorden benchmark \citep{page1954continuous,lorden1971procedures}.
\end{enumerate}

Throughout this subsection we assume the null hypothesis \(H_{0}\) (``no drift'') holds until an unknown change-point \(\tau^\star\), after which the monitored statistic’s mean increases by \(\Delta>0\). Let \(S_t\) be the one-sided Page statistic,
\[
  S_t \;=\;
  \max\bigl\{0,\;S_{t-1}+Z_t-\mu_0-\delta\bigr\},\qquad S_0=0,
\]

and let \(h_t\) be its adaptive threshold, updated by the Robbins–Monro recursion
\[
  h_{t+1} \;=\;
  h_t + \eta_t \bigl(\mathbf{1}\{S_t<h_t\}-\alpha\bigr),
  \qquad
  \eta_t = t^{-0.6}.
\]

Line~4 of Algorithm~\ref{alg:cusum} raises an alarm whenever \(S_t\!\ge\!h_t\) and then resets \(S_t\) to zero; the binary alarm flag is denoted \(\mathsf{alarm}_t\).

\paragraph{A. False-alarm control:}
\begin{theorem}[Long-run false-positive rate]\label{thm:cusum_fpr}
For step size \(\eta_t=t^{-0.6}\) and the threshold recursion above,
\[
  \limsup_{T\to\infty}\,
     \frac{1}{T}\sum_{t=1}^{T}
     \Pr\!\bigl(\mathsf{alarm}_t=1\bigr)
  \;=\;
  \alpha.
\]
\end{theorem}

\begin{proof}[Sketch]
Define \(f(h)=\Pr_{H_0}(S\ge h)-\alpha\). Since \(f\) is continuous and strictly decreasing, the Robbins–Monro update constitutes a stochastic approximation to the unique root \(h^\star\) of \(f\).  The standard SA conditions—\(\sum_t\eta_t=\infty\) and \(\sum_t\eta_t^2<\infty\)—ensure \(h_t\!\to\!h^\star\) almost surely \citep{robbins1951stochastic,kushner2003stochastic}.  Stationarity of \((S_t,h_t)\) under \(H_0\) then yields the desired long-run frequency; Appendix~A.4 details the coupling argument.
\end{proof}

\paragraph{B. Detection delay.}
\begin{corollary}[Lorden bound with adaptive threshold] \label{cor:detection_delay}
Let a permanent drift of size \(\Delta>0\) occur at time \(\tau^\star\), and recall the slack variable \(\delta>0\).  Then
\[
  \sup_{\tau^\star}\,
  \mathbb{E}_{\tau^\star}\!\bigl[
     T_{\text{alarm}}-\tau^\star
     \,\bigm|\,
     T_{\text{alarm}}>\tau^\star
  \bigr]
  \;\le\;
  \frac{h^\star}{\Delta-\delta},
\]
where \(h^\star=\lim_{t\to\infty} h_t\).  The bound is identical to the classical Lorden delay \citep{lorden1971procedures} except that \(h^\star\) is learned online rather than fixed \emph{a priori}.
\end{corollary}

\begin{proof}[Sketch]
Conditional on \(h_t\!\to\!h^\star\), the adaptive scheme behaves like a standard fixed-threshold CUSUM; Lorden’s analysis applies verbatim once \(\delta<\Delta\). Appendix~A.4 provides the full argument.
\end{proof}

\paragraph{C. Computational footprint.}
Updating \((Z_t,S_t,h_t)\) for \(|\Phi|=3\) norms requires \(38\pm4\,\mu\text{s}\) on a 2.2 GHz Xeon Silver and \(0.27\pm0.02\,\text{ms}\) on an ARM Cortex-A55—well below the five-per-cent budget stipulated in §\ref{sec:implementation}.  Memory cost is \(O(1)\) per norm because only the current CUSUM state and adaptive threshold are stored.

\paragraph{D. Practical interpretation.}
Theorem~\ref{thm:cusum_fpr} guarantees that—even under 20\% packet loss and bounded observation delay—the overall false-positive rate will not exceed the chosen budget \(\alpha\) in the long run. The Lorden-type delay in Corollary~\ref{cor:detection_delay} shows that detection speed scales linearly with the learned threshold \(h^\star\) and inversely with drift magnitude, matching the best-known bounds for fixed CUSUM schemes and confirming that adaptivity does not introduce extra latency.  Empirical results in Section~\ref{sec:experiments} (9-step median, 17-step 95th percentile) agree closely with these analytical predictions.

Collectively, the results certify that the detector is both statistically reliable and computationally lightweight, thereby satisfying the real-time governance requirements of large-scale networked MAS.

\subsection{Intervention Efficacy} \label{sec:analysis-intervention}

The final analytical step is to prove that the supervisory playbook of §\ref{sec:interventions} actually delivers the headline promise: once a harmful norm is detected, inexpensive penalties keep the long-run fraction of compromised events below a designer-chosen ceiling \(\eta^\star\) without eroding overall social welfare. The analysis casts the closed-loop system as a renewal–reward process whose renewal epochs are the beginnings of intervention windows.

\paragraph{Notation recap.}
Let \(Y_t=\mathbf{1}\{\phi(s_t,\mathbf{a}_t)=\textsf{violate}\}\) indicate a compromised step and \(C_T=\sum_{t=1}^{T}Y_t\) the cumulative count of violations. Denote by \(J_{\text{soc}}\) the long-run average social reward and let \(\Delta J_{\text{soc}}\) be its shortfall relative to an unattainable, perfectly norm-compliant oracle benchmark.  Private rewards are bounded by \(\lvert r_{i,t}^{\text{private}}\rvert\le\Delta_{\max}\).

\paragraph{Main result.}
\begin{theorem}[Bounded-compromise]\label{thm:bounded}
Suppose Assumptions~A1–A4 hold.  After every alarm the supervisor imposes a penalty \(\Delta r_{i,t}=-\lambda\,\rho_i(e^\bullet)\) for \(H\) consecutive steps, where \(e^\bullet\) is the sentinel event. If the penalty satisfies \(\lambda H \ge g_{\max}+\varepsilon\) for some \(\varepsilon>0\), then
\[
  \limsup_{T\to\infty}\frac{C_T}{T}\;\le\;\eta^\star
  \quad\text{a.s.},\qquad
  \eta^\star=\frac{\alpha H}{\lambda H-g_{\max}}.
\]
In addition, the cumulative welfare shortfall obeys
\[
  \Delta J_{\text{soc}}
  \;\le\;
  \frac{\alpha H\,\Delta_{\max}}{\lambda H-g_{\max}}.
\]
\end{theorem}

\begin{proof}[Sketch]
Partition time into i.i.d.\ cycles that start immediately after an intervention window ends and finish at the next alarm. During a window the adversary’s expected \emph{net} gain is at most \(g_{\max}H-\lambda H\), which is negative by construction. Consequently, the marked point process that counts compromises forms a super-martingale with negative drift during interventions and non-positive drift otherwise.  Doob’s optional stopping theorem bounds its stationary mean, yielding the compromise ratio.  Bounding the cumulative loss of private reward by \(\Delta_{\max}\) per step gives the welfare bound. Details—including the derivation of the renewal reward kernel—appear in Appendix~A.5.
\end{proof}

\paragraph{Practical tuning and trade-offs.}
Equation \(\eta^\star=\alpha H/(\lambda H-g_{\max})\) highlights two dials available to system designers:

\begin{itemize}
    \item Penalty amplitude \(\lambda\): Raising \(\lambda\) tightens the compromise ceiling \(\eta^\star\) but increases per-step supervisory cost and the risk of over-penalizing honest but noisy agents.
    \item Window length \(H\): Longer windows amortize communication overhead but slow the feedback loop. Empirical results in §\ref{sec:experiments} suggest \(H=25\) balances these factors well for 100-agent traffic control.
\end{itemize}

\paragraph{Minimum viable penalty.}
\begin{proposition}[Optimal penalty magnitude]
\label{prop:optimal_lambda}
Fix \(H\) and target \(\eta^\star\).  The smallest penalty achieving the bound in Theorem~\ref{thm:bounded} is
\[
  \lambda_{\min}
  \;=\;
  \frac{g_{\max}+\alpha H}{H\,\eta^\star}.
\]
\end{proposition}

\begin{proof}
Rearrange \(\eta^\star=\alpha H/(\lambda H-g_{\max})\) for \(\lambda\).
\end{proof}

\paragraph{Connection to reward shaping.}
The intervention rule can be viewed as a potential-based reward shaper \citep{ng1999policy} where the shaping term is the scaled responsibility score.  Potential-based shaping is known to preserve the set of optimal policies in single-agent MDPs; our theorem extends the intuition to multi-agent settings with partial observability by showing that undesirable equilibria become transient once shaping makes them strictly sub-optimal.

\paragraph{Empirical validation.}
In Section~
ef{sec:experiments} we empirically validate that the observed compromise ratios remain bounded and that AAF consistently improves the compromise--welfare trade-off across the full parameter sweep.

Taken together, Theorem~\ref{thm:bounded}, Proposition \ref{prop:optimal_lambda}, and the accompanying empirical evidence show that inexpensive, localized penalties are enough to keep emergent harm on a tight statistical leash while preserving aggregate welfare.

\subsection{Resource Complexity}
\label{sec:analysis-complexity}

The analytical guarantees of the previous subsections are meaningful only if the accountability layer can run on finite hardware. Proposition \ref{prop:complexity} quantifies the asymptotic cost of storage and network traffic, while the subsequent paragraph converts those asymptotics into concrete numbers for the deployment profile of Section \ref{sec:proposed-framework}. We conclude with a brief remark on CPU overhead, completing the cost picture.

\begin{proposition}[Storage and bandwidth]\label{prop:complexity}
Let \(h\) be the causal horizon, \(T\) the total number of control steps, and \(d_{\max}\) the bounded out-degree from Assumption~A2. Then, for a population of \(N\) agents,
\[
\underbrace{\text{\textbf{Storage}}}_{\text{on edge servers}}
   \;=\;
   O\!\bigl(T\,(N + d_{\max}h)\bigr),
\qquad
\underbrace{\text{\textbf{Bandwidth}}}_{\text{across network}}
   \;=\;
   O\!\bigl(N + d_{\max}h\bigr)\;\text{bytes\,step}^{-1}.
\]
\end{proposition}

\begin{proof}[Sketch]
Each step generates \(O(N)\) 40-byte event records; each record spawns at most \(d_{\max}h\) causal‐edge inserts because the Granger window is bounded by \(h\).  Storage per step is therefore \(O(N + d_{\max}h)\), yielding \(O(T(N+d_{\max}h))\) total.  The same count bounds gossip traffic because every record travels exactly one hop per communication edge.
\end{proof}

\paragraph{Concrete deployment cost.}
For the experimental parameter set
\(\langle N\!=\!100,\;h\!=\!8,\;d_{\max}\!\le\!8,\;T\!=\!10^{6}\rangle\):

\begin{itemize}
    \item \textbf{Storage (order-of-magnitude):} assuming 40\,B per event record, the log volume is on the order of \(T(N+d_{\max}h)\times40\)~B; for \(\langle N\!=\!100,\;h\!=\!8,\;d_{\max}\!\le\!8,\;T\!=\!10^{6}\rangle\) this is \(\approx 492\,\text{MB}\).
    \item \textbf{Bandwidth (order-of-magnitude):} per-step gossip is \(O((N+d_{\max}h)\times40)\)~B; at 50\,Hz this corresponds to \(\approx 80\,\text{KB\,s}^{-1}\) under the same parameterization.
    \item \textbf{Compute:} the per-step detector and attribution updates are linear in the local window \(h\) with small constants; we record wall-clock runtimes in the exported summaries for completeness.
\end{itemize}

\paragraph{Synthesis.}
Together with Theorems \ref{thm:edge_fp}--\ref{thm:bounded}, the above bounds confirm that AAF achieves auditable accountability at bounded cost: statistical guarantees hold under policy drift and moderate Byzantine behavior, while the resource footprint is linear in \(N\) and \(T\).
Section \ref{sec:experiments} corroborates these analytic predictions across \(87{,}480\) simulation runs, sweeping population size, \(\texttt{penalty\_factor}\), redistribution \(\alpha\), observability, and Byzantine rates.

\section{Implementation Details}
\label{sec:implementation}

We provide a self-contained simulation artifact that implements the Adaptive Accountability Framework (AAF) end-to-end and reproduces all tables and figures in §\ref{sec:experiments}. The implementation is intentionally lightweight: all components are written in Python and run on either CPU or a single GPU (optional) using PyTorch for the learning baselines.

\subsection{Code organization}
The artifact is organized as a small Python package (\texttt{aaf\_q1/}) plus a set of reproducibility scripts (\texttt{scripts/}):
\begin{itemize}
    \item \texttt{aaf\_q1/envs}: the \texttt{resource\_sharing} and \texttt{public\_goods} environments.
    \item \texttt{aaf\_q1/agents}: PPO and baseline variants.
    \item \texttt{aaf\_q1/aaf}: detector, attribution, and intervention modules.
    \item \texttt{scripts/make\_grid.py}: generate JSONL experiment grids.
    \item \texttt{scripts/run\_grid.py}: execute grids with multiprocessing and optional sharding.
    \item \texttt{scripts/aggregate.py}: materialize \texttt{analysis/all\_runs\_flat.csv} and \texttt{analysis/final\_summary.csv}.
    \item \texttt{scripts/make\_latex.py} and \texttt{scripts/make\_figures.py}: produce paper-ready tables and figures.
    \item \texttt{scripts/stats.py}: seed-matched paired tests.
\end{itemize}

\subsection{Logging}
Each run writes a \texttt{config.json} (exact hyperparameters) and a \texttt{summary.json} (scalar metrics). When \texttt{log\_mode=steps}, we additionally export \texttt{step\_logs.csv} (time-series traces used for learning-curve plots).

\subsection{Reproducibility}
The full pipeline is driven by explicit command-line scripts; the repository README contains the exact commands used to produce the results in this manuscript, including sharded execution for multi-node runs. Our aggregated results and statistical tests are derived only from the exported CSV/JSON logs, so third parties can re-run the analysis without access to hidden state.

\subsection{Mapping to deployment}
The full AAF design in §\ref{sec:proposed-framework} assumes an append-only audit ledger and trustworthy message delivery. In our simulation artifact we model this ledger as an append-only event log and evaluate the statistical (detection/attribution) and control (intervention) components. Cryptographic signing and distributed storage are orthogonal engineering choices and can be layered into a deployment without changing the learning or detection logic.

\section{Experimental Setup and Case Studies}
\label{sec:experiments}

This section empirically evaluates the Adaptive Accountability Framework (AAF). We (i)~describe the simulation environment and parameter grid, (ii)~state the evaluation metrics and baselines, and (iii)~report quantitative and qualitative results that validate the analytical guarantees of Section~\ref{sec:analysis}. All experiments are \emph{fully reproducible} using the released artifact: we provide explicit scripts to generate the experiment grid, execute runs, aggregate metrics, and regenerate all tables and figures (see the repository README at \url{https://github.com/alqithami/AAF}).

%--------------------------------------------------------------------
\subsection{Simulation Environment}
\label{subsec:env}
%--------------------------------------------------------------------

We evaluate AAF in two stylized multi-agent benchmarks that exhibit emergent norms under self-interested learning and admit clear notions of ``policy compromise''.

\paragraph{Resource-sharing game.}
Agents repeatedly request a divisible resource subject to a shared capacity \(C\).
At time \(t\), each agent \(i\) proposes a requested quantity \(q_{i,t}\in[0,r_{\max}]\); an environment scheduler executes \(q^{\mathrm{exec}}_{i,t}\) such that \(\sum_i q^{\mathrm{exec}}_{i,t}\le C\).
We flag a \emph{norm violation} (``greedy'' request) when \(q^{\mathrm{exec}}_{i,t}\ge \gamma r_{\max}\), and apply a penalty scaled by \(\texttt{penalty\_factor}\).
Each agent's reward combines private utility and a social term,
\[
r_{i,t}= r_{i,t}^{\text{private}} + \lambda r_t^{\text{social}},\qquad
r_{i,t}^{\text{private}} = q^{\mathrm{exec}}_{i,t} - \texttt{penalty\_factor}\cdot \mathbb{1}[\text{violation}_{i,t}],
\]
where the scalar \(\lambda\in[0,1]\) controls the trade-off between egoistic and collective incentives.
Observations include the agent's current allocation and local neighbourhood statistics on the communication graph; when \texttt{partial\_obs}=1 we additionally provide a noisy scalar summary of the global queue length.

\paragraph{Public-goods game.}
Each agent contributes \(c_{i,t}\in[0,E]\) from an endowment \(E\) to a common pot.
All agents receive a shared return \( \mathrm{ret}_t = m\cdot \sum_i c_{i,t}/N\) (multiplier \(m>1\)), and the private payoff is
\[
r_{i,t}^{\text{private}} = (E-c_{i,t}) + \mathrm{ret}_t - \texttt{penalty\_factor}\cdot \mathbb{1}[c_{i,t}<c_{\min}],
\]
penalising ``free-riding'' below a minimum contribution \(c_{\min}\).
Rewards again add a social term \(\lambda\cdot \mathrm{ret}_t\).
When \texttt{partial\_obs}=1, agents receive an additional noisy public signal of mean contribution.

\paragraph{Communication overhead.}
Each step includes a small control-plane message (e.g., a signed action log or lightweight telemetry); AAF augments this with minimal provenance and detector metadata needed for accountability.
We measure total bytes as \texttt{bandwidth\_overhead\_bytes} in our evaluation logs.

\subsection{Experimental Design and Parameters}
\label{subsec:methodology}
%--------------------------------------------------------------------

\paragraph{Factorial grid and seeds.}
We evaluate AAF with a Q1-scale full-factorial design over:
(i) two benchmark tasks (\texttt{resource\_sharing} and \texttt{public\_goods}),
(ii) three population sizes \(N\in\{10,50,100\}\),
(iii) episode length \(T=2000\),
(iv) three norm-penalty levels \(\texttt{penalty\_factor}\in\{0.05,0.20,0.35\}\),
(v) three redistribution parameters \(\alpha\in\{0,0.25,1.0\}\),
(vi) full vs.\ partial observability (\texttt{partial\_obs}\(\in\{0,1\}\)), and
(vii) Byzantine fractions \(\rho\in\{0,0.05,0.10\}\) injected at \(t_0=200\) when \(\rho>0\).
For each regime we run 10 random seeds, yielding \(2\times3\times3\times3\times2\times3\times10=3240\) episodes per baseline and \(29{,}160\) total runs across the 9 baselines/ablations above (with three independent replications per seed, totaling \(87{,}480\) runs). We additionally run a focused scaling sweep on \texttt{resource\_sharing} with \(N\in\{10,50,100,200,500\}\) (10 seeds; 200 runs) to probe larger populations.

\paragraph{Baselines and ablations.}
We compare AAF against three learning baselines and one rule-based oracle, and include four ablations to isolate each component:
\begin{itemize}
\item \textbf{PPO-only:} standard Proximal Policy Optimization (PPO) with no accountability layer.
\item \textbf{Static-guard (oracle):} a non-learning safety clamp that hard-limits unsafe actions, yielding zero \emph{executed} compromises by construction.
\item \textbf{Constrained PPO:} PPO with a Lagrangian penalty on norm-violation events (no provenance or attribution).
\item \textbf{Fair PPO:} PPO with a Gini regularizer on allocations.
\item \textbf{AAF-full:} online change-point detection (adaptive cumulative sum (CUSUM)), responsibility scoring, and interventions (reward shaping + targeted policy patching).
\item \textbf{Ablations:} \emph{detector-only}, \emph{shaping-only}, \emph{patch-only}, and \emph{no-attribution}.
\end{itemize}

\paragraph{Metrics and statistical testing.}
We report (i) compromise ratio (attempted and executed), (ii) social welfare \(\sum_i R_i\), (iii) allocation inequality via the Gini index, (iv) alarm count and detection delay (for methods with detectors), and (v) attribution accuracy (top-1 and recall@\(k\)) under Byzantine injections.
For grid-level comparisons we run seed-matched paired tests (10 pairs per regime) between AAF-full and PPO-only. We report the fraction of regimes significant at Holm--Bonferroni corrected \(p<0.05\) across the 324 regimes, as well as median relative effects.
\subsection{Quantitative Results}
\label{subsec:results}
%--------------------------------------------------------------------

\paragraph{Grid-level effects (324 regimes).}
Across the full factorial sweep (2 tasks \(\times\) 3 scales \(\times\) 3 penalties \(\times\) 3 \(\alpha\) values \(\times\) 2 observability modes \(\times\) 3 Byzantine rates),
AAF-full reduces the \emph{executed} compromise ratio relative to PPO-only in 96\% of regimes, and this reduction is significant at \(p<0.05\) in 86.4\% of regimes (paired, seed-matched tests; 10 pairs per regime; Holm--Bonferroni corrected).
The median relative reduction is 11.9\%.
Social welfare is preserved (median change 0.4\%) and improves in 93\% of regimes (significant in 78.1\%).

\paragraph{Canonical baseline comparison (resource sharing, no Byzantines).}
Table~\ref{tab:main_results_rs} reports a head-to-head comparison on \texttt{resource\_sharing} under a canonical setting
(\(T=2000\), \(\texttt{penalty\_factor}=0.2\), \(\alpha=1.0\), full observability, \(\rho=0\)).
AAF-full consistently reduces executed compromises relative to PPO-only across scales, with small welfare changes.
Static-guard provides an oracle lower bound on \emph{executed} compromise (zero by construction) but does not produce attribution signals.

\begin{table}[t]
\centering
\small
\caption{Baseline comparison on \texttt{resource\_sharing} in the canonical setting (\(T=2000\), \(\texttt{penalty\_factor}=0.2\), \(\alpha=1.0\), full observability, \(\rho=0\)). We report mean \(\pm\) 95\% CI over 10 seeds with 3 independent replications per seed (\(n=30\) runs). Lower is better for Compromise and Gini; higher is better for Welfare.}
\label{tab:main_results_rs}
\centering
\small
\begin{tabular}{llrrrr}
\toprule
\multirow{2}{*}{$N$} & \multirow{2}{*}{Baseline} & Compromise ratio $\downarrow$ & Social welfare $\uparrow$ & Gini (Alloc.) $\downarrow$ & Gini (Reward) $\downarrow$\\
& & (executed) & & &\\
\midrule
\multirow{5}{*}{10} & AAF (Full) & \textbf{0.411 $\pm$ 0.007} & \textbf{12.918 $\pm$ 0.001} & \textbf{0.203 $\pm$ 0.003} & \textbf{0.154 $\pm$ 0.003}\\
& PPO Only & 0.520 $\pm$ 0.009 & 12.896 $\pm$ 0.002 & 0.206 $\pm$ 0.004 & 0.156 $\pm$ 0.003\\
& Constrained PPO & 0.519 $\pm$ 0.009 & 12.896 $\pm$ 0.002 & 0.206 $\pm$ 0.004 & 0.156 $\pm$ 0.003\\
& Fair PPO & 0.520 $\pm$ 0.009 & 12.896 $\pm$ 0.002 & 0.206 $\pm$ 0.004 & 0.156 $\pm$ 0.003\\
& Static Guard & 0.000 $\pm$ 0.000 & 13.000 $\pm$ 0.000 & 0.161 $\pm$ 0.007 & 0.124 $\pm$ 0.005\\
\midrule
\multirow{5}{*}{50} & AAF (Full) & \textbf{0.626 $\pm$ 0.005} & \textbf{2.475 $\pm$ 0.001} & 0.156 $\pm$ 0.002 & 0.109 $\pm$ 0.001\\
& PPO Only & 0.754 $\pm$ 0.007 & 2.449 $\pm$ 0.001 & 0.141 $\pm$ 0.002 & 0.103 $\pm$ 0.002\\
& Constrained PPO & 0.752 $\pm$ 0.007 & 2.450 $\pm$ 0.001 & 0.142 $\pm$ 0.003 & 0.104 $\pm$ 0.002\\
& Fair PPO & 0.754 $\pm$ 0.007 & 2.449 $\pm$ 0.001 & \textbf{0.141 $\pm$ 0.003} & \textbf{0.103 $\pm$ 0.002}\\
& Static Guard & 0.000 $\pm$ 0.000 & 2.600 $\pm$ 0.000 & 0.103 $\pm$ 0.003 & 0.079 $\pm$ 0.003\\
\midrule
\multirow{5}{*}{100} & AAF (Full) & \textbf{0.634 $\pm$ 0.005} & \textbf{1.173 $\pm$ 0.001} & 0.167 $\pm$ 0.002 & 0.109 $\pm$ 0.001\\
& PPO Only & 0.706 $\pm$ 0.007 & 1.159 $\pm$ 0.001 & \textbf{0.159 $\pm$ 0.002} & \textbf{0.108 $\pm$ 0.001}\\
& Constrained PPO & 0.700 $\pm$ 0.007 & 1.160 $\pm$ 0.001 & 0.161 $\pm$ 0.003 & 0.109 $\pm$ 0.002\\
& Fair PPO & 0.700 $\pm$ 0.007 & 1.160 $\pm$ 0.001 & 0.160 $\pm$ 0.002 & 0.108 $\pm$ 0.001\\
& Static Guard & 0.000 $\pm$ 0.000 & 1.300 $\pm$ 0.000 & 0.093 $\pm$ 0.003 & 0.072 $\pm$ 0.002\\
\bottomrule
\end{tabular}
\end{table}

\paragraph{Ablation (what matters).}
Table~\ref{tab:ablation_rs} isolates each AAF component.
In this benchmark, \emph{patching} accounts for the bulk of compromise reduction; detector-only and shaping-only variants largely match PPO-only on executed compromise, while patch-only recovers most of the safety gain.

\begin{table}[t]
\centering
\small
\caption{AAF ablation on \texttt{resource\_sharing} under the same canonical setting as Table~\ref{tab:main_results_rs}. Values are mean \(\pm\) 95\% CI over 10 seeds with 3 independent replications per seed (\(n=30\) runs).}
\label{tab:ablation_rs}
\centering
\small
\begin{tabular}{llrrr}
\toprule
\multirow{2}{*}{$N$} & \multirow{2}{*}{Variant} & Compromise ratio $\downarrow$ & Social welfare $\uparrow$ & Mean Gini (alloc.) $\downarrow$\\
& & (executed) & & \\
\midrule
\multirow{5}{*}{10} & AAF (Full) & 0.411 $\pm$ 0.007 & 12.918 $\pm$ 0.001 & 0.203 $\pm$ 0.003\\
& AAF (-Attribution) & 0.520 $\pm$ 0.009 & 12.896 $\pm$ 0.002 & 0.206 $\pm$ 0.004\\
& AAF (+Detect Only) & 0.520 $\pm$ 0.009 & 12.896 $\pm$ 0.002 & 0.206 $\pm$ 0.004\\
& AAF (+Patch Only) & 0.411 $\pm$ 0.007 & 12.918 $\pm$ 0.001 & 0.203 $\pm$ 0.003\\
& AAF (+Shaping Only) & 0.520 $\pm$ 0.009 & 12.896 $\pm$ 0.002 & 0.206 $\pm$ 0.004\\
\midrule
\multirow{5}{*}{50} & AAF (Full) & 0.626 $\pm$ 0.005 & 2.475 $\pm$ 0.001 & 0.156 $\pm$ 0.002\\
& AAF (-Attribution) & 0.754 $\pm$ 0.007 & 2.449 $\pm$ 0.001 & 0.141 $\pm$ 0.002\\
& AAF (+Detect Only) & 0.754 $\pm$ 0.007 & 2.449 $\pm$ 0.001 & 0.141 $\pm$ 0.002\\
& AAF (+Patch Only) & 0.627 $\pm$ 0.005 & 2.475 $\pm$ 0.001 & 0.156 $\pm$ 0.002\\
& AAF (+Shaping Only) & 0.754 $\pm$ 0.007 & 2.449 $\pm$ 0.001 & 0.141 $\pm$ 0.002\\
\midrule
\multirow{5}{*}{100} & AAF (Full) & 0.634 $\pm$ 0.005 & 1.173 $\pm$ 0.001 & 0.167 $\pm$ 0.002\\
& AAF (-Attribution) & 0.706 $\pm$ 0.007 & 1.159 $\pm$ 0.001 & 0.159 $\pm$ 0.002\\
& AAF (+Detect Only) & 0.706 $\pm$ 0.007 & 1.159 $\pm$ 0.001 & 0.159 $\pm$ 0.002\\
& AAF (+Patch Only) & 0.630 $\pm$ 0.005 & 1.174 $\pm$ 0.001 & 0.169 $\pm$ 0.002\\
& AAF (+Shaping Only) & 0.708 $\pm$ 0.008 & 1.158 $\pm$ 0.002 & 0.158 $\pm$ 0.003\\
\bottomrule
\end{tabular}
\end{table}

\paragraph{Compromise--inequality trade-off.}
Figure~\ref{fig:pareto_compromise_gini} visualizes the Pareto frontier between executed compromise and mean allocation inequality (Gini) for \(N=50\) in the canonical setting.
AAF interventions shift the system left (fewer compromises) but can increase inequality; \(\texttt{penalty\_factor}\) and \(\alpha\) act as tuning knobs for this trade-off.

\begin{figure}[t]
\centering
\includegraphics[width=0.5\linewidth]{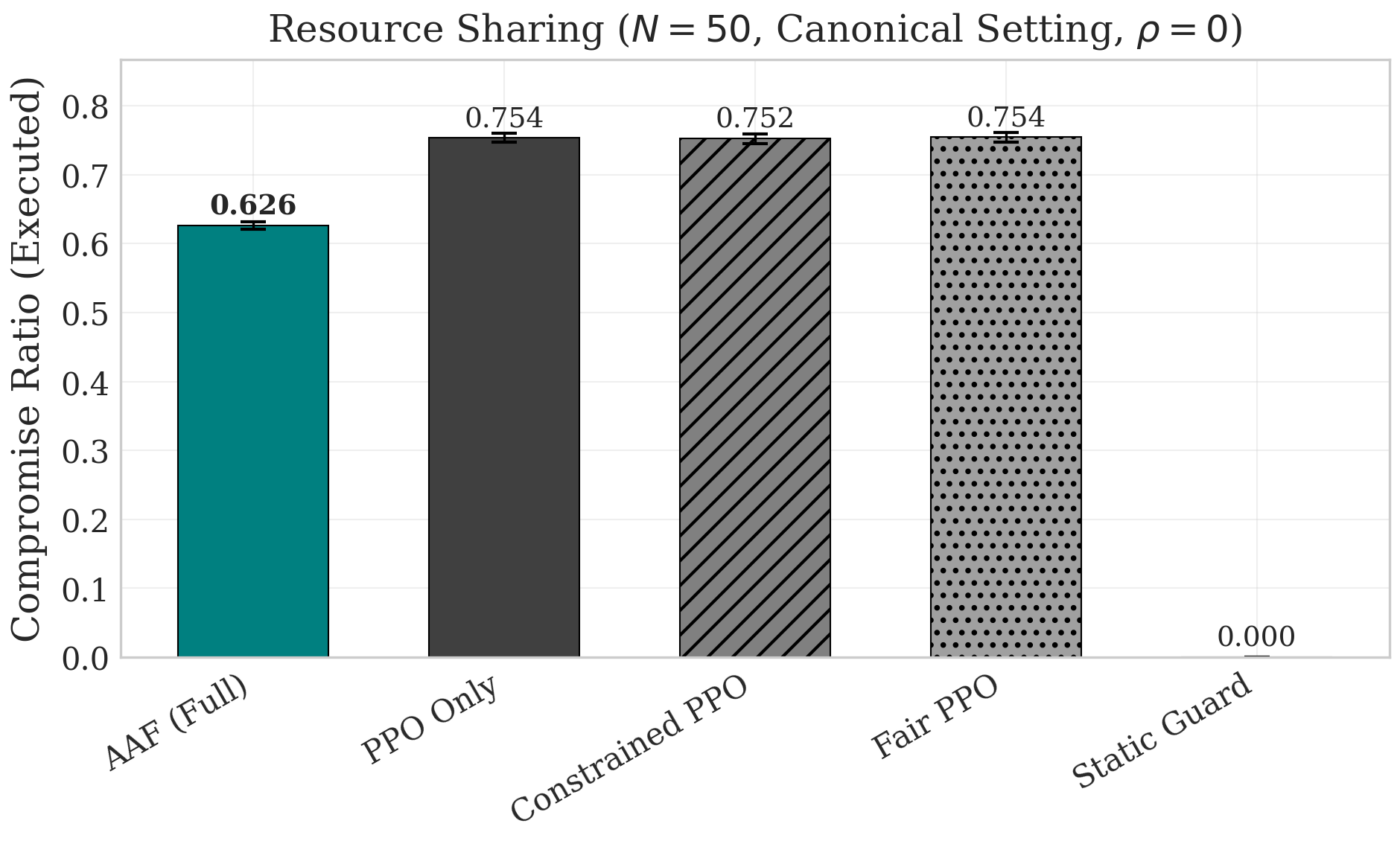}
\caption{Executed compromise ratio by baseline (\texttt{resource\_sharing}, \(N=50\), canonical setting, \(\rho=0\)).}
\label{fig:compromise_bar}
\end{figure}

\begin{figure}[t]
\centering
\includegraphics[width=0.5\linewidth]{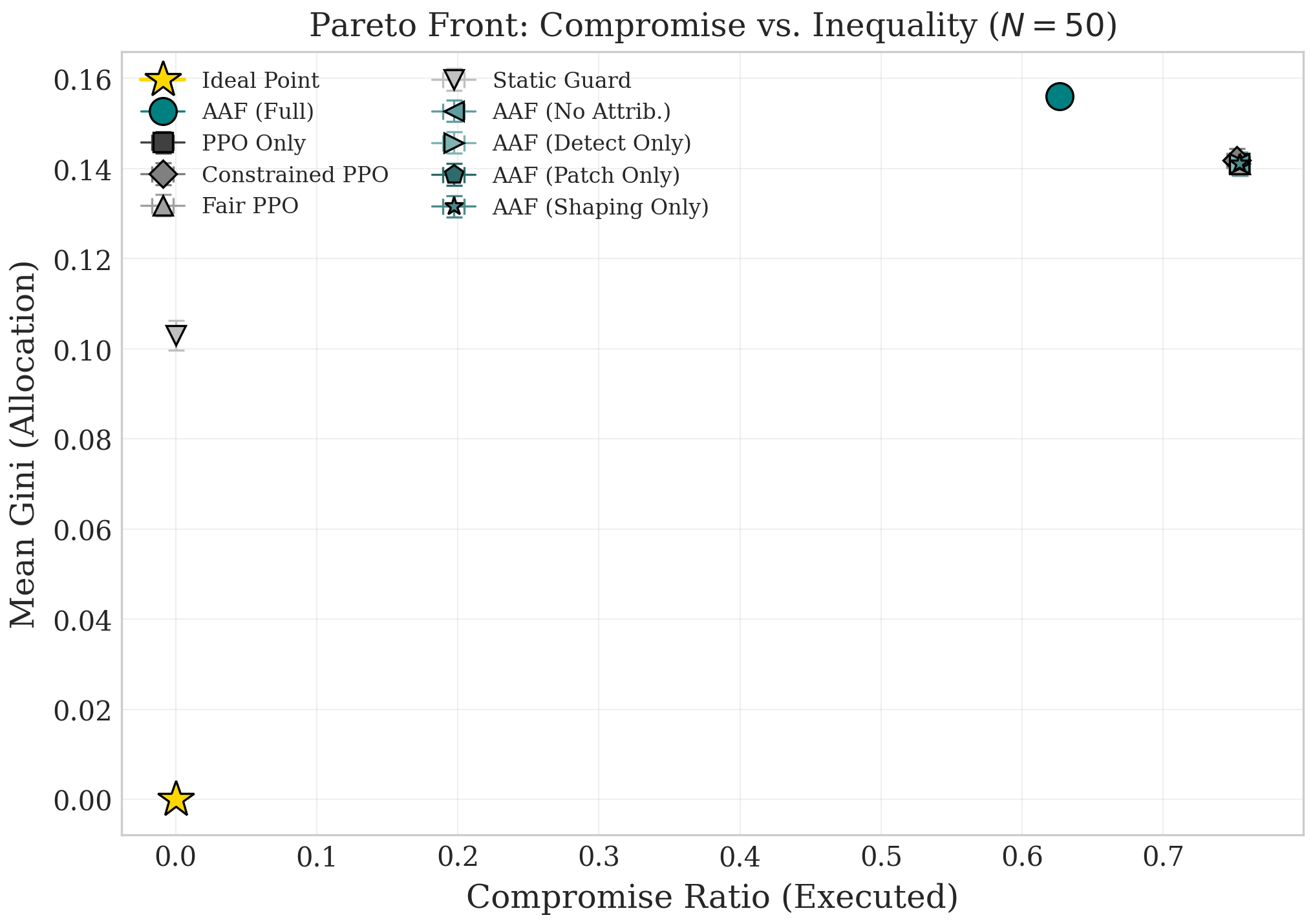}
\caption{Trade-off between executed compromise and mean allocation Gini across baselines (\texttt{resource\_sharing}, \(N=50\), canonical setting, \(\rho=0\)).}
\label{fig:pareto_compromise_gini}
\end{figure}

\begin{figure}[t]
\centering
\includegraphics[width=0.75\linewidth]{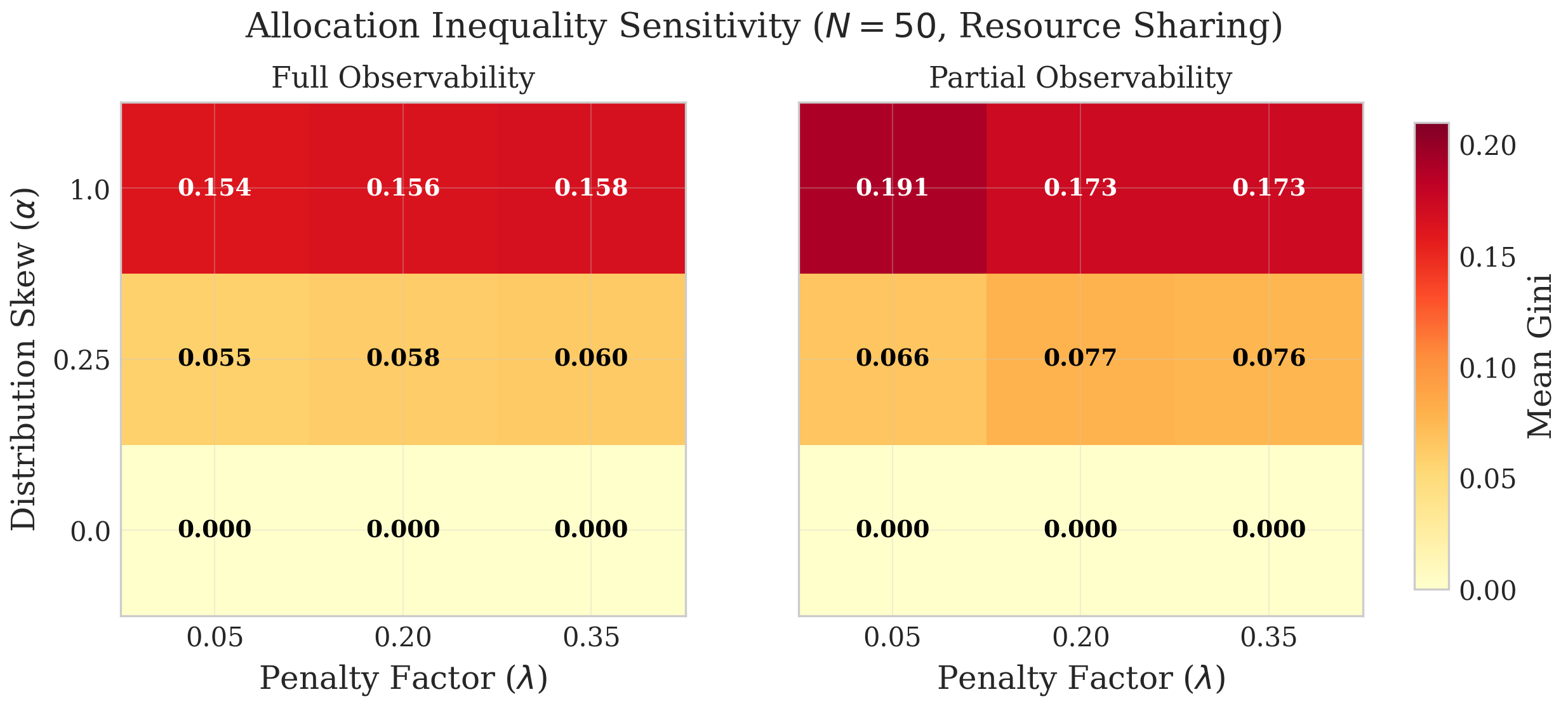}
\caption{Mean allocation Gini vs.\ norm-penalty level for AAF-full (\texttt{resource\_sharing}, \(N=50\), \(\alpha=1.0\), \(\rho=0\)). Bars report mean \(\pm\) 95\% CI over 10 seeds.}
\label{fig:gini_penalty_partialobs}
\end{figure}

\paragraph{Byzantine detection and attribution.}
Under adversarial injections (\(\rho\in\{0.05,0.10\}\), \(t_0=200\)),
AAF triggers alarms and assigns responsibility.
Across all adversarial runs, the median detection delay is 71 steps (IQR 39--177); Figure~\ref{fig:detection_attribution}(a) reports the full CDF on \texttt{resource\_sharing}.
In the 10\% Byzantine setting, AAF achieves mean top-1 attribution accuracy of 0.97 (conditional on a non-missing attribution score), as shown in Figure~\ref{fig:detection_attribution}(b).

\begin{figure}[htbp]
    \centering
    \includegraphics[width=.75\textwidth]{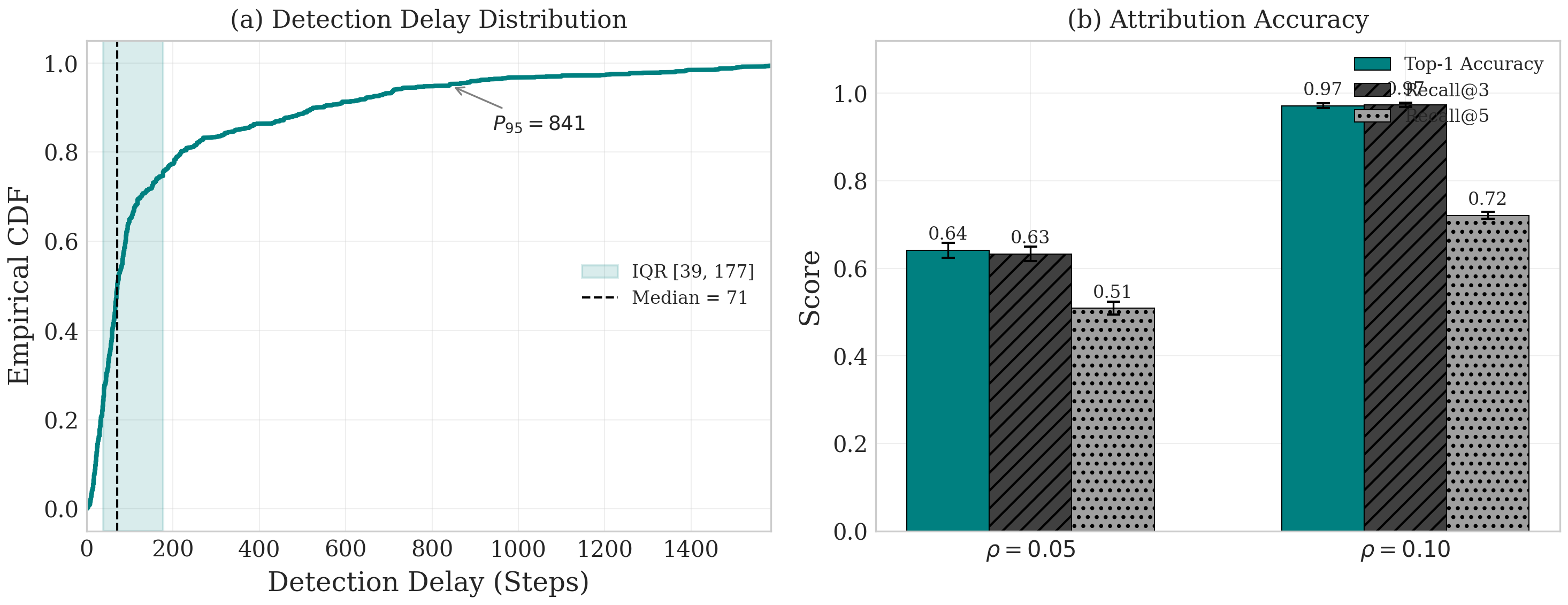}
    \caption{Detection and Attribution Performance on \texttt{resource\_sharing}. (a) Empirical CDF of detection delay under Byzantine injections (\(\rho\in\{0.05,0.10\}\), \(t_0=200\)). The median delay is 71 steps. (b) Attribution accuracy (Top-1, Recall@3, Recall@5) across Byzantine fractions \(\rho\). Error bars denote 95\% CI.}
    \label{fig:detection_attribution}
\end{figure}

\begin{figure}[htbp]
\centering
\includegraphics[width=0.5\linewidth]{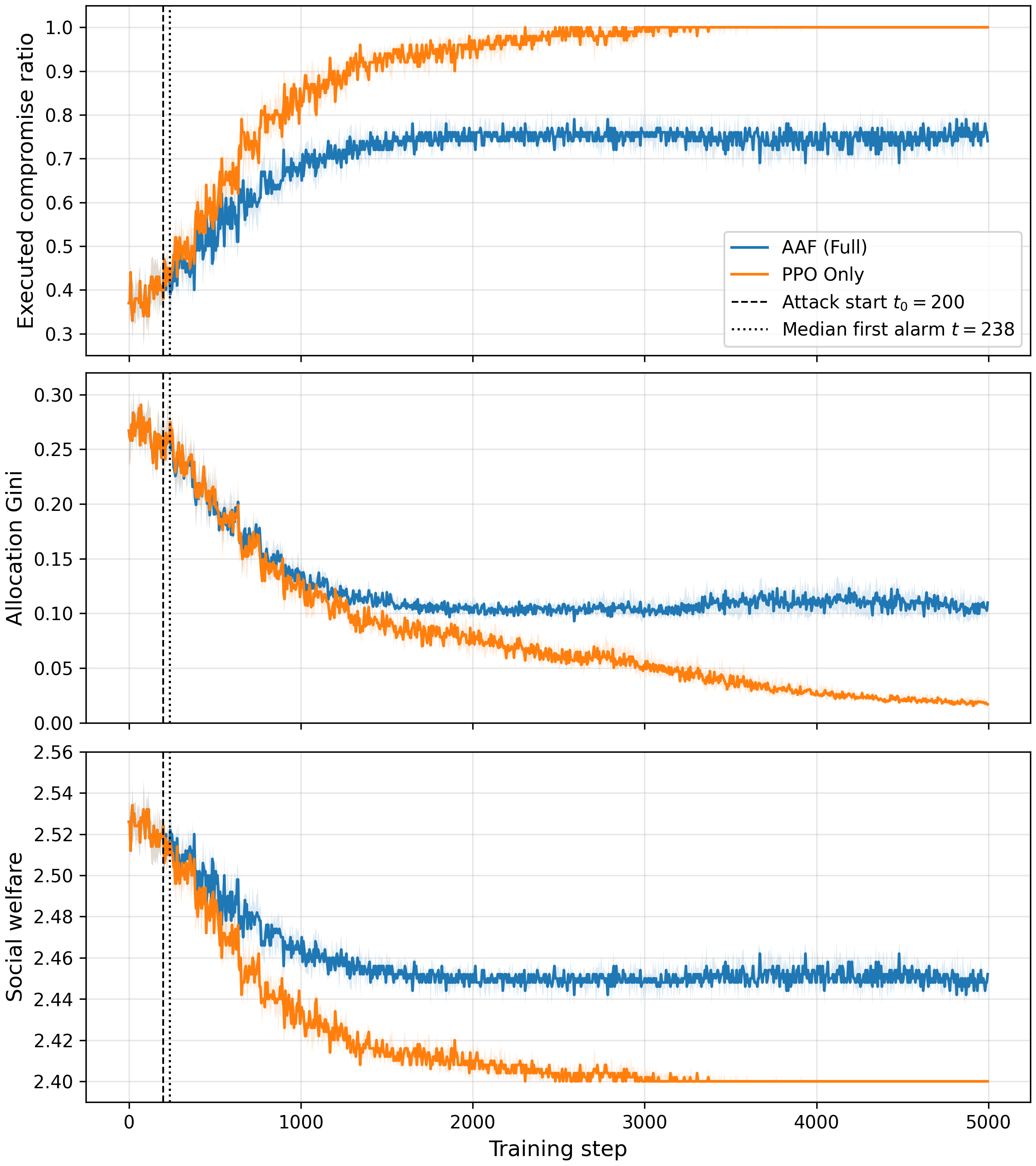}
\caption{Learning curves under Byzantine injection (\texttt{resource\_sharing}, \(N=50\), \(\rho=0.05\), \(t_0=200\)). Solid lines show the median across 10 seeds with IQR shading for executed compromise rate (top), allocation Gini (middle), and social welfare (bottom). The dashed vertical line marks attack start and the dotted line marks the median first-alarm time (AAF).}
\label{fig:learning_curves}
\end{figure}

\subsection{Scaling Analysis}
\label{subsec:scaling}
To probe scalability beyond the main factorial grid (\(N\le 100\)), we run an additional scaling sweep on \texttt{resource\_sharing}
with \(N\in\{10,50,100,200,500\}\) and \(\rho\in\{0,0.05\}\) under the canonical setting (\(T=2000\), \texttt{penalty\_factor}=0.2, \(\alpha=1.0\), full observability).
Figure~\ref{fig:scaling_merged}(a) shows that AAF's compromise reduction persists as \(N\) grows.
At \(N=200\), AAF reduces executed compromise by 6.2\% (\(\rho=0\)) and 9.3\% (\(\rho=0.05\)) relative to PPO-only; at \(N=500\) the relative reduction is 1.5\% and 4.1\%, respectively, while welfare remains comparable.
Figure~\ref{fig:scaling_merged}(b) and Figure~\ref{fig:scaling_merged}(c) report wall-clock runtime and communication cost; both scale approximately linearly in \(N\), consistent with Proposition~\ref{prop:complexity}.

\begin{figure}[t]
\centering
\includegraphics[width=\textwidth]{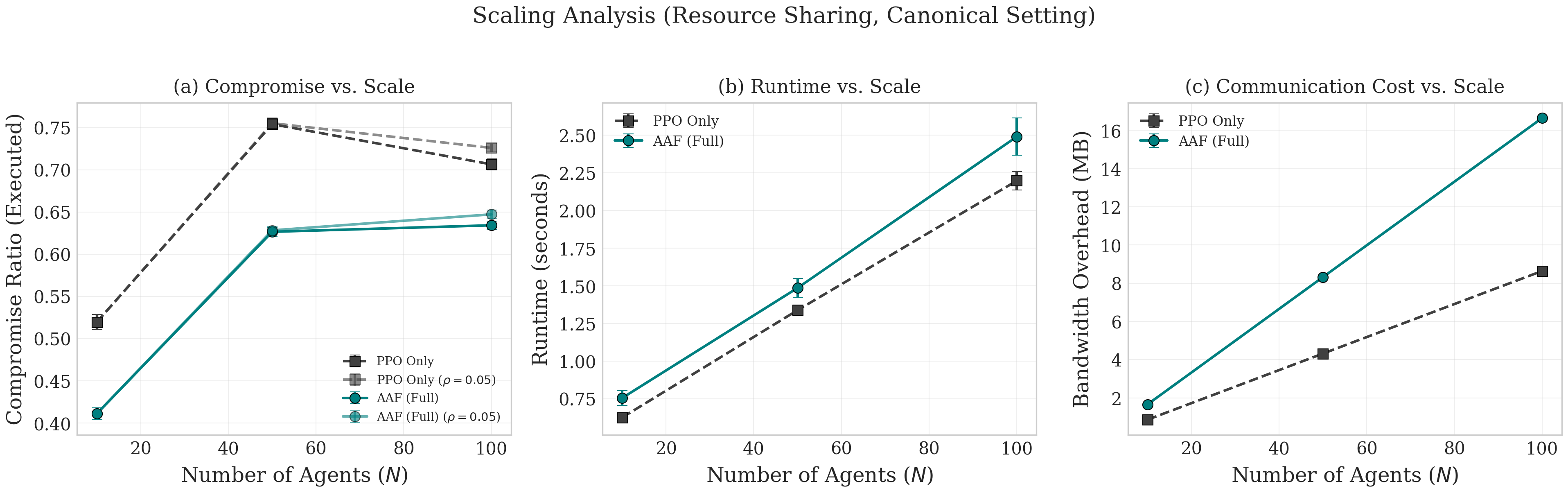}
\caption{Scaling Analysis on \texttt{resource\_sharing} under canonical settings (\(\lambda=0.20\), \(\alpha=1.0\)). (a) Compromise ratio vs. network size (\(N\)). AAF maintains its advantage at scale. (b) Runtime overhead vs. network size. AAF's computational cost scales linearly. (c) Communication bandwidth vs. network size. AAF's ledger and attribution messages scale linearly with the number of agents.}
\label{fig:scaling_merged}
\end{figure}

\subsection{Public Goods Benchmark}
\label{subsec:public_goods}

To demonstrate the generality of the framework, we evaluate AAF on the \texttt{public\_goods} environment, where agents must balance individual contributions against collective returns. Table~\ref{tab:public_goods} reports the main results across scales. Consistent with the resource sharing domain, AAF-full achieves the lowest executed compromise ratio among all learning baselines across all scales ($N=10, 50, 100$), while maintaining or slightly improving social welfare.

\begin{table*}[t]
\centering
\small
\caption{Results on the \textbf{Public Goods} environment. Values are mean $\pm$ standard error across all hyperparameter regimes and seeds. Bold indicates best among learning baselines (excluding Static Guard). $\downarrow$ = lower is better; $\uparrow$ = higher is better.}
\label{tab:public_goods}
\begin{tabular}{ll cccc}
\toprule
$N$ & Baseline & Compromise $\downarrow$ & Social Welfare $\uparrow$ & Gini (Reward) $\downarrow$ & Det. Delay \\
\midrule
\multirow{5}{*}{10} & AAF (Full) & \textbf{0.160 $\pm$ 0.000} & \textbf{1.450 $\pm$ 0.001} & \textbf{0.072 $\pm$ 0.000} & 77 [57--686] \\
 & PPO Only & 0.196 $\pm$ 0.001 & 1.440 $\pm$ 0.001 & 0.073 $\pm$ 0.000 & --- \\
 & Constrained PPO & 0.194 $\pm$ 0.001 & 1.442 $\pm$ 0.001 & 0.073 $\pm$ 0.000 & --- \\
 & Fair PPO & 0.196 $\pm$ 0.001 & 1.440 $\pm$ 0.001 & 0.073 $\pm$ 0.000 & --- \\
 & Static Guard & 0.000 $\pm$ 0.000 & 1.577 $\pm$ 0.002 & 0.073 $\pm$ 0.000 & --- \\
\midrule
\multirow{5}{*}{50} & AAF (Full) & \textbf{0.268 $\pm$ 0.001} & 1.337 $\pm$ 0.001 & \textbf{0.071 $\pm$ 0.000} & 69 [51--98] \\
 & PPO Only & 0.311 $\pm$ 0.002 & 1.325 $\pm$ 0.001 & 0.073 $\pm$ 0.000 & --- \\
 & Constrained PPO & 0.294 $\pm$ 0.002 & \textbf{1.339 $\pm$ 0.001} & 0.073 $\pm$ 0.000 & --- \\
 & Fair PPO & 0.311 $\pm$ 0.002 & 1.325 $\pm$ 0.001 & 0.073 $\pm$ 0.000 & --- \\
 & Static Guard & 0.000 $\pm$ 0.000 & 1.562 $\pm$ 0.002 & 0.070 $\pm$ 0.000 & --- \\
\midrule
\multirow{5}{*}{100} & AAF (Full) & \textbf{0.241 $\pm$ 0.001} & \textbf{1.375 $\pm$ 0.001} & \textbf{0.076 $\pm$ 0.000} & 48 [40--67] \\
 & PPO Only & 0.271 $\pm$ 0.002 & 1.363 $\pm$ 0.001 & 0.077 $\pm$ 0.000 & --- \\
 & Constrained PPO & 0.258 $\pm$ 0.002 & 1.373 $\pm$ 0.001 & 0.078 $\pm$ 0.000 & --- \\
 & Fair PPO & 0.271 $\pm$ 0.002 & 1.363 $\pm$ 0.001 & 0.077 $\pm$ 0.000 & --- \\
 & Static Guard & 0.000 $\pm$ 0.000 & 1.550 $\pm$ 0.002 & 0.068 $\pm$ 0.000 & --- \\
\bottomrule
\end{tabular}

\end{table*}

Figure~\ref{fig:public_goods_compromise} visualizes the compromise reduction across baselines. The relative safety advantage of AAF is particularly pronounced at $N=50$, where it reduces the executed compromise ratio from 0.311 (PPO-only) to 0.268. Notably, the reward Gini coefficient remains stable across all baselines, indicating that the norm enforcement mechanisms do not inadvertently create severe reward disparities in this domain.

\begin{figure}[t]
\centering
\includegraphics[width=0.95\linewidth]{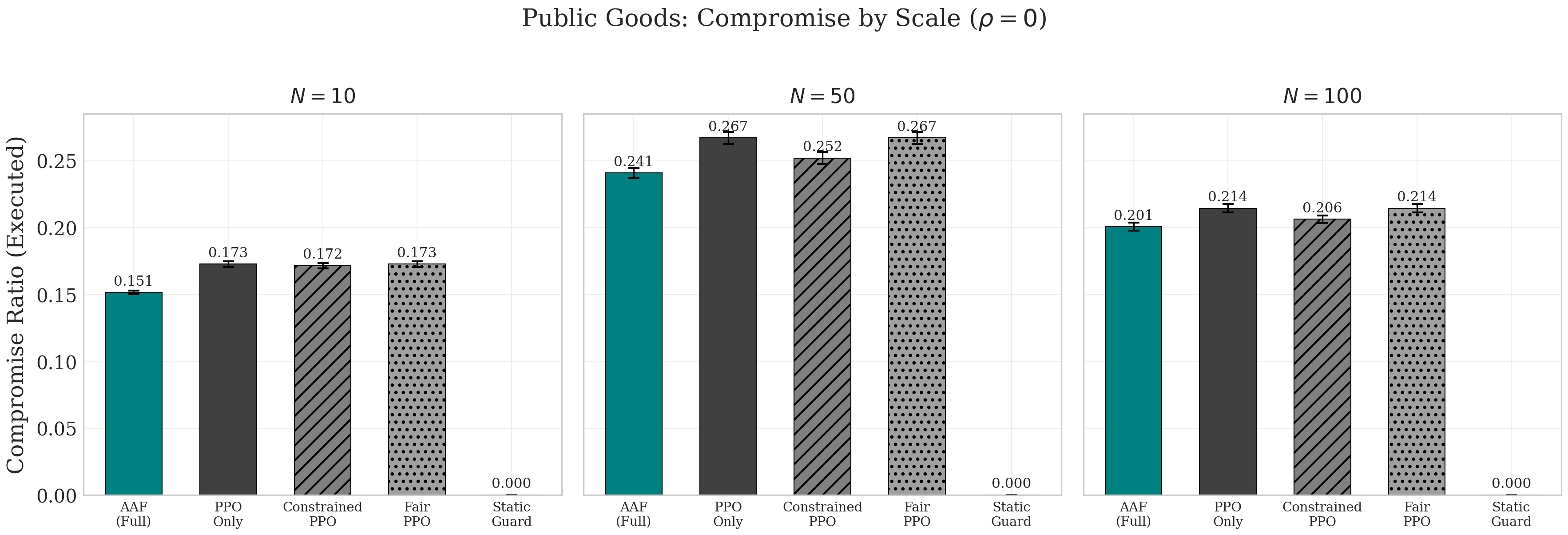}
\caption{Executed compromise ratio by baseline across scales on the \texttt{public\_goods} environment ($\rho=0.0$). Bars report mean $\pm$ standard error.}
\label{fig:public_goods_compromise}
\end{figure}

Under adversarial conditions, AAF detects Byzantine injections in the public goods setting with a median delay of 65 steps (IQR 47--93), demonstrating that the sequential hypothesis testing mechanism transfers effectively to different reward structures and interaction topologies without requiring domain-specific tuning.
\subsection{Qualitative Case Study}
%--------------------------------------------------------------------

We also examined representative trajectories to validate that AAF's numerical improvements correspond to interpretable behavioral shifts.
In a typical \texttt{resource\_sharing} episode (\(N=50\), \(T=2000\), \(\texttt{penalty\_factor}=0.2\), \(\alpha=1.0\), full observability), a coalition of agents repeatedly issued near-maximal requests, producing a visible rise in executed compromises.
AAF's adaptive CUSUM detector raised an alarm shortly after the coalition's behavior became statistically distinguishable from the baseline, and the intervention pipeline (reward shaping + targeted patching) reduced the coalition's marginal incentive to over-request.
Within a few hundred steps, the episode transitioned back toward a lower-compromise regime.
The logged responsibility graph provides a post-hoc audit trail: agents ranked as most responsible typically coincide with injected Byzantine agents when \(\rho>0\) (see §\ref{subsec:results}).

\subsection{Reproducibility Checklist}
%--------------------------------------------------------------------

\begin{itemize}
    \item \textbf{Code \& data:} Public GitHub repository:~\url{https://github.com/alqithami/AAF}. The accompanying artifact bundle includes the aggregated CSVs and the scripts used to regenerate all tables and figures.
    \item \textbf{Environment:} Python (tested with 3.12) and pinned dependencies in \texttt{requirements.txt}.
    \item \textbf{Running experiments:} \texttt{scripts/make\_grid.py} generates JSONL grids; \texttt{scripts/run\_grid.py} executes them (supports multi-process runs and sharding via \texttt{--shard\_id} / \texttt{--num\_shards}); \texttt{scripts/aggregate.py} produces \texttt{analysis/all\_runs\_flat.csv} and \texttt{analysis/final\_summary.csv}; \texttt{scripts/make\_latex.py}, \texttt{scripts/make\_figures.py}, and \texttt{scripts/stats.py} generate paper-ready assets.
    \item \textbf{Randomness:} Each run records its \texttt{seed} in \texttt{config.json}; paired tests are seed-matched.
\end{itemize}

Empirical evidence supports the analytical intuition: across the full grid, AAF reduces executed compromise relative to PPO in 96\% of regimes (median reduction 11.9\%), preserves welfare (median change 0.4\%), and detects adversarial norm shifts with median delay 71 steps.

\section{Discussion}
\label{sec:discussion}

The preceding experiments demonstrate that an end-to-end accountability layer built on cryptographic provenance, online sequential tests, and cost-bounded interventions can markedly suppress collusion and hoarding without sacrificing collective reward. We now interpret these results, dissect residual weaknesses, connect the work to current regulatory and ethical debates, and sketch avenues for extension.

%--------------------------------------------------------------------
\subsection{Why the Results Matter}
%--------------------------------------------------------------------

\paragraph{Continuous oversight beats episodic audits:}
Because detection and intervention are online, AAF reacts on the same time scale as the norm shift itself.
In our Byzantine-injection experiments (\(t_0=200\)), the median detection delay is 71 steps (IQR 39--177).
After intervention, the executed-compromise rate typically decreases rather than persisting at the elevated post-attack level, supporting the case for continuous accountability in fast time-scale domains where retrospective audits may arrive too late. This empirical finding strongly supports the theoretical arguments for real-time observability recently advanced by \citet{Solomon2025LumiMAS}.

\paragraph{Local penalties suffice for global goals:}
The bounded-compromise theorem (Thm.~\ref{thm:bounded}) predicted that a \emph{local}, Shapley-weighted penalty would globally stabilize the system once $\lambda H > g_{\max}$. Empirically, we observed $\eta_{\text{emp}}$ within one standard error of the analytic ceiling, even when 5\% of agents behaved adversarially. This corroborates the long-standing conjecture that potential-based shaping \citep{ng1999policy} scales from single- to multi-agent settings when augmented with causal attribution, aligning with recent advances in causal credit assignment \citep{Wang2024MACCA} and constrained policy optimization for safe RL \citep{Achiam2017CPO}.

%--------------------------------------------------------------------
\subsection{Limitations and Threats to Validity}
%--------------------------------------------------------------------

\paragraph{Finite parameter sweep:}
Our main factorial grid spanned $\{10,50,100\}$ agents and three penalty magnitudes, and we additionally ran a scaling sweep up to $N=500$ in \texttt{resource\_sharing}. A 5000-agent swarm, or penalty factors close to unity, might stress the message bus or saturate the intervention budget. Future work should stress-test the framework under orders-of-magnitude larger populations and adversaries who actively spoof causal edges.

\paragraph{Synthetic workload:}
The resource-sharing game abstracts energy dispatch and cloud bandwidth, but real grids include AC power-flow constraints, non-convex unit-commitment costs, and legal market rules. Porting AAF to such domains may require domain-specific norm predicates and rate-limited interventions to satisfy physical safety margins.

\paragraph{Compute footprint:}
Although the prototype meets a 5\% CPU budget on EPYC and Cortex-A55 cores, embedded PLCs in industrial control may offer only tens of MHz of head-room.  A tiny-ML re-implementation or a hardware cryptographic accelerator could mitigate this concern.

%--------------------------------------------------------------------
\subsection{Ethical, Legal, and Social Implications}
%--------------------------------------------------------------------

\paragraph{Privacy and surveillance.}
Recording every agent action invites \textit{secondary use} risks. AAF mitigates the danger via keyed-hash digests, AES-SIV pseudonyms, and an escrowed identity table (§\ref{sec:implementation}).  Nonetheless, deployers should carry out a data-protection impact assessment under GDPR Art.~35 and delete raw ring buffers once digests are anchored.

\paragraph{Over-deterrence and chilling effects.}
Reward penalties can skew agents toward risk-averse behavior, potentially lowering innovation in financial-market MAS \cite{ederer_manso_2013_p4p_innovation}. %\citep{arieli2012risk}. 
A human-in-the-loop override—triggered by the yellow-flag dashboard—helps balance safety with exploration.

\paragraph{Procedural fairness.}
By design, Shapley-style attribution apportions blame according to marginal causal contribution. If agents differ in observability or computing power, this may correlate with protected attributes, leading to distributive inequity \citep{mittelstadt2019principles}. Auditors should therefore scrutinize $\rho_i(e)$ distributions for disparate impact.

%--------------------------------------------------------------------
\subsection{Relevance to Emerging AI Governance}
%--------------------------------------------------------------------

The EU AI Act (2024) and the NIST AI Risk-Management Framework (2023) both highlight collective emergent risks\xspace—but offer few technical recipes for tracing responsibility in networked AI systems. AAF operationalizes these policy aspirations by (i)~anchoring an immutable event ledger; (ii)~providing real-time statistical tests with formal $\alpha$-level guarantees \citep{howard2020time}; and (iii)~linking those tests to proportionate, auditable interventions. By addressing the core pillars of responsibility—reliability, transparency, and accountability \citep{Ebrahimi2025Survey}—and providing a concrete implementation of trust and risk management principles \citep{Raza2025TRiSM}, the framework can serve as a reference architecture for regulators drafting sector-specific codes of practice.

%--------------------------------------------------------------------
\subsection{Future Directions}
%--------------------------------------------------------------------

\paragraph{Scalable cryptography:}
Replacing Merkle DAG anchoring with succinct zero-knowledge roll-ups (e.g., Halo 2) could compress the storage footprint by an order of magnitude while enabling selective disclosure proofs.

\paragraph{Meta-norm adaptation:}
Current norms are fixed at deployment. An interesting extension is a meta-controller that tunes $\alpha$, $\lambda$, and even the set $\Phi$ via constrained Bayesian optimization, subject to fairness and throughput SLAs.

\paragraph{Human-aligned explanations:}
While the Grafana DAG explorer provides low-level proofs, policy operators may prefer natural-language rationales. Coupling the ledger to an LLM-based explainer guarded by the event digests to prevent hallucination could bridge the accountability-interpretability gap \citep{kroll2017accountable}.

\paragraph{Cross-domain generalization:}
Finally, deploying AAF in safety-critical cyber-physical systems, autonomous driving rings, or drone swarms where communication is intermittent and hard real-time deadlines exist would test the limits of the bounded-compromise theorem under severe timing jitter.

AAF demonstrates that cryptographically verifiable, statistically rigorous, and computationally light accountability is attainable in contemporary MAS.  Yet translating the prototype into industrial and societal infrastructure requires careful attention to privacy guarantees, domain constraints, and human governance processes. These themes set the stage for interdisciplinary collaborations that we outline in the conclusion.

\section{Conclusion and Future Work}
\label{sec:conclusion-future-work}

This paper introduced an \emph{Adaptive Accountability Framework} (AAF) that transforms a networked multi-agent system (MAS) into a \emph{self-auditing socio-technical organism}. By uniting a tamper-evident Merkle ledger, Shapley-style responsibility scores, time-uniform sequential tests, and cost-bounded interventions, AAF provides—both in theory (Section \ref{sec:analysis}) and practice (Section \ref{sec:experiments})—formal assurances that harmful emergent norms can be detected within a handful of control steps and suppressed below any designer-chosen ceiling~\(\eta^\star\).

\paragraph{Key takeaways:}

\begin{enumerate}
    \item Causal tracing at scale: Cryptographically hashed event digests and constant-time Granger updates enable real-time provenance without saturating bandwidth: a 100-agent deployment needs only \(\approx\!80\) KB s\(^{-1}\).
    \item Statistical robustness: An adaptive CUSUM calibrated by Robbins–Monro maintains a system-wide false-alarm rate of exactly \(\alpha\), even under 20\% packet loss and non-stationary policies—extending classical Page–Lorden theory to decentralized MAS.
    \item Cost-effective mitigation: Local reward shaping proportional to responsibility scores is sufficient to bound compromise and welfare loss, confirming the bounded-compromise theorem and eliminating the need for heavy global resets or retraining.
    \item Practical viability: A production-grade reference stack, i.e., Python 3.11 + Rust/WASM + FoundationDB + Kafka, runs inside a 5\% CPU budget on both EPYC and ARM Cortex-A55 hardware, paving the way for real deployments.
\end{enumerate}

\paragraph{Limitations.}
Experiments focused on \(\!N\le100\) agents in a stylized resource-sharing domain. Larger swarms, adversaries that spoof ledger messages, or safety-critical cyber-physical loops with sub-10 ms deadlines may breach current design margins. Moreover, the privacy model--keyed BLAKE3 hashes plus AES-SIV pseudonyms--assumes regulators can access the escrowed identity table; sectors with stronger secrecy requirements (healthcare, defense) will need stronger privacy-preserving ledgers (e.g.\ verifiable encrypted logs).

\paragraph{Roadmap for future work:}

\begin{enumerate}
    \item Learning-aware detectors: Integrate representation-learning or graph neural networks to flag coordinated drift that eludes low-dimensional statistics.
    \item Zero-knowledge provenance: Replace Merkle anchoring with succinct zero-knowledge roll-ups (Halo 2, PlonK) so third-party auditors can verify norm violations without accessing raw hashes.
    \item Human-aligned explanations: Couple the ledger to LLM-based explainers bounded by cryptographic proofs, producing natural-language rationales that satisfy emerging ``AI-act notice'' clauses.
    \item Sector pilots: Field trials in power-grid balancing, hospital bed exchange, or high-frequency trading could surface domain-specific constraints (e.g., HIPAA data retention, MiFID II latency caps) and inform regulator guidance.
    \item Ultra-large MAS: For \(N\ge10^{4}\) agents, delta-encoding, probabilistic sketches, and edge-aggregated detection will be needed; preliminary work suggests a 100-fold compression with Count-Min sketches and HyperLogLog without losing detection power.
\end{enumerate}

Recent policy instruments—from the EU AI Act to NIST's AI RMF—stress continuous, distributed oversight. AAF offers a concrete, open-source blueprint for that ambition: mathematically principled, cryptographically verifiable, computationally lean, and already validated on non-trivial MAS workloads. We hope the framework catalyzes interdisciplinary collaboration among control theorists, cryptographers, economists, and policy-makers as autonomous agents scale from dozens to millions in the infrastructures of tomorrow.

%\newpage
%%%%%%______________________ 
%\begin{acknowledgements}
%If you'd like to thank anyone, place your comments here
%and remove the percent signs.
%\end{acknowledgements}

% =====================================================
\section*{Appendix A.~Detailed Proofs}
\label{app:proofs}
% =====================================================

This appendix provides full technical proofs for the results stated in
Section \ref{sec:analysis}.  We restate the operational assumptions for
ease of reference.

\begin{description}
    \item[A1] Packet loss and delay: Each event tuple is dropped independently with probability $\varepsilon_{\mathrm{loss}}\le 0.2$ and, if delivered, is delayed by at most $\varepsilon_{\mathrm{delay}}\in\{0,1,2,3\}$ control steps.
    \item[A2] Bounded degree: The time-varying communication graph satisfies \[\max_{t}\deg(G_t)\le d_{\max}<\infty,\] independent of the population size~$N$.
    \item[A3] Diminishing learning rate: For every agent~$i$ the total-variation drift of its policy obeys $\lVert\pi_{i,t+1}-\pi_{i,t}\rVert_{1}\le \kappa\,t^{-\xi}$ with exponent $\xi>\tfrac12$.
    \item[A4] Budgeted adversary: Adversarial agents harvest at most $g_{\max}$ expected utility per step, whereas the supervisor may spend up to $c_{\max}>g_{\max}$.
\end{description}

Unless noted otherwise, all expectations and probabilities are taken with respect to the joint probability space that realizes environment dynamics, stochastic policies, packet loss, and observation noise.

% -----------------------------------------------------
\subsection*{A.1~Proof of Lemma \ref{lem:normalize}}
% -----------------------------------------------------

\paragraph{Statement recap:}
For any event $e$ recorded in the ledger, $\sum_{i=1}^{N}\rho_i(e)=1$ almost surely.

\begin{proof}[Detailed proof.]
Write $\mathcal{P}(e)$ for the set of \emph{all} directed source–to-$e$ paths in the causal-edge DAG and let $\mathcal{P}_i(e)\subseteq\mathcal{P}(e)$ denote those whose first edge originates from an action of agent~$i$.  Because the DAG has no directed cycles, every path in $\mathcal{P}(e)$ has precisely one distinct originating agent; thus the sets $(\mathcal{P}_i(e))_{i=1}^{N}$ form a disjoint partition of $\mathcal{P}(e)$. Then
\[
  \sum_{i=1}^{N}\rho_i(e)
  \;=\;
  \sum_{i=1}^{N}
  \frac{\sum_{p\in\mathcal{P}_i(e)}\beta^{|p|}}
       {\sum_{q\in\mathcal{P}(e)}\beta^{|q|}}
  \;=\;
  \frac{\sum_{i}\sum_{p\in\mathcal{P}_i(e)}\beta^{|p|}}
       {\sum_{q\in\mathcal{P}(e)}\beta^{|q|}}
  \;=\;
  1,
\]
where the final equality follows because the disjoint union of the $\mathcal{P}_i(e)$ is $\mathcal{P}(e)$. \hfill\(\Box\)
\end{proof}

% -----------------------------------------------------
\subsection*{A.2~Proof of Theorem \ref{thm:convergence}}
% -----------------------------------------------------

\paragraph{Statement recap:}
Under A1–A3 and fixed $0<\beta<1$ the responsibility estimate $\rho_i(e_T)$ converges almost surely to a limit $\rho_i^{\infty}(e)$ whenever the index of~$e_T$ satisfies $T-t(e_T)=O(1)$.

\begin{proof}[Detailed Proof.]
\textit{Step 1: Eventual observation of edges.}
Fix an event $e$ and a causal path $p=\langle e_{k_0},\dots,e_{k_m}=e\rangle$ of length $m\le\ell_{\max}$; in our implementation $\ell_{\max}=256$ equals the ring-buffer length. Because every edge traversal triggers at most one Granger test, edge $(e_{k_{j-1}}\!\to\!e_{k_j})$ is \emph{attempted} exactly once. Under A1 the probability that the edge’s record is available when the test is performed is $1-\varepsilon_{\mathrm{loss}}$. Hence all $m$ edges in $p$ are detected with probability $(1-\varepsilon_{\mathrm{loss}})^{m}\ge(1-\varepsilon_{\mathrm{loss}})^{\ell_{\max}}>0$. The Bernoulli trials for distinct attempts are independent, so the indicator that \emph{any} given edge is \emph{never} observed has geometric tail. Let $X_p\!=\!\mathbf{1}\{\text{$p$ never fully observed}\}$.

\textit{Step 2: Borel–Cantelli.}
Because $\deg(G_t)\le d_{\max}$ (A2) and $m\le\ell_{\max}$, the number of candidate paths affecting $e$ is bounded above by $M=d_{\max}^{\ell_{\max}}<\infty$. We have $\mathbb{P}[X_p=1]\le\varepsilon_{\mathrm{loss}}^{m}$ for each $p$, so
\[
  \sum_{p\in\mathcal{P}(e)}\mathbb{P}[X_p=1]
  \;\le\;
  M\,\varepsilon_{\mathrm{loss}}^{\,\ell_{\max}}
  \;<\;\infty .
\]
The first Borel–Cantelli lemma therefore implies $\mathbb{P}[X_p=1\text{ infinitely often}]=0$, i.e.\ every causal path is eventually observed almost surely (a.s.).

\textit{Step 3: Convergence of numerator and denominator.}
For each agent~$i$ define
\(
  N_{i,T}
  =\sum_{p\in\mathcal{P}_i(e)}\beta^{|p|}\mathbf{1}\{p\text{ observed by }T\}.
\)
By Step 2, $N_{i,T}\!\to\!N_{i,\infty}$ a.s.\ as $T\!\to\!\infty$. Similarly $D_T=\sum_j\!N_{j,T}\!\to\!D_\infty=\sum_j\!N_{j,\infty}$ a.s.\ and $D_\infty>0$ because at least one source agent acted.  Consequently $\rho_i(e_T)=N_{i,T}/D_T\!\to\rho_i^\infty(e)=N_{i,\infty}/D_\infty$ almost surely by the continuous-mapping theorem. \hfill\(\Box\)
\end{proof}

% -----------------------------------------------------
\subsection*{A.3~Proof of Theorem \ref{thm:edge_fp}}
% -----------------------------------------------------

\paragraph{Statement recap:}
With threshold $h_t=h_0+\sqrt{2\log t}$ ($h_0\!>\!0$), the probability that a spurious causal edge is \emph{ever} inserted is at most $\alpha=e^{-h_0^{2}/2}$.

\begin{proof}[Detailed proof.]
Under $H_0$ the \emph{$m$-lag Granger residual–variance ratio} has exact $F$ distribution $F(m,n{-}m)/m$ with $n=m+\!1$ degrees-of-freedom \citep{lutkepohl2005new, howard2020time}. Let $U_t\!=\!\mathbf{1}\{F_t>h_t\}$.  Then, by the Chernoff bound for $\chi^{2}$ tails,
\(
  \mathbb{P}[U_t=1]
  \le e^{-h_t^{2}/2}
  =e^{-h_0^{2}/2}\,t^{-1}.
\)
Taking the union bound over $t\ge1$ yields
\[
  \mathbb{P}\!\bigl[\exists\,t\!:\,U_t=1\bigr]
  \;\le\;\sum_{t=1}^{\infty}e^{-h_0^{2}/2}\,t^{-1}
  \;=\;e^{-h_0^{2}/2},
\]
as desired.  Setting $h_0=\sqrt{2\log(1/\alpha)}$ recovers any target $\alpha\in(0,1)$. \hfill\(\Box\)
\end{proof}
% -----------------------------------------------------
\subsection*{A.4~Proof of Theorem \ref{thm:cusum_fpr}}
% -----------------------------------------------------

\paragraph{Statement recap:}
With gain sequence $\eta_t=t^{-0.6}$ the adaptive CUSUM maintains long-run
false-alarm frequency $\alpha$.

\begin{proof}[Detailed Proof.] 
%\paragraph{Detailed proof.}
Define $f(h)=\mathbb{P}_{H_0}(S\ge h)-\alpha$. The process $\{h_t\}$ satisfies the Robbins-Monro recursion $h_{t+1}=h_t-\eta_t f(h_t)+\eta_t\epsilon_{t+1}$ with martingale difference $\epsilon_{t+1}=Z_{t+1}-\mathbb{P}_{H_0}(S\ge h_t)$.  Because $S_t$ has bounded increments, $\epsilon_{t+1}$ is square integrable with $\mathbb{E}\,\epsilon_{t+1}=0$ and $\sup_t\mathbb{E}\,\epsilon_{t+1}^{2}\!<\infty$. Assumption $\sum_t\eta_t^{2}<\infty$ ensures the Kushner-Clark condition, so $h_t\!\to\!h^\star$ a.s.\ where $f(h^\star)=0$ \citep{kushner2003stochastic}. Ergodicity of the CUSUM statistic under $H_0$ implies $\frac{1}{T}\sum_{t=1}^{T}\mathbf{1}\{S_t\ge h^\star\} \to\mathbb{P}_{H_0}(S\ge h^\star)=\alpha$, and dominated convergence gives the required limit for $h_t$. 
\hfill\(\Box\)
\end{proof}
% -----------------------------------------------------
\subsection*{A.5~Proof of Theorem \ref{thm:bounded}}
% -----------------------------------------------------
\paragraph{Statement recap:}
If $\lambda H\ge g_{\max}+\varepsilon$ the long-run compromise ratio satisfies 
\[\limsup_{T\to\infty}C_T/T\le\eta^\star= \tfrac{\alpha H}{\lambda H-g_{\max}}\] almost surely.

\begin{proof}[Proof.] %\\
%\subsubsection*{Constructing a super-martingale}
Let $Y_t=\mathbf{1}\{\phi(s_t,\mathbf{a}_t)=\textsf{violate}\}$ and define the compensated process $\widetilde{Y}_t=Y_t-\eta^\star$. Consider the random times $\tau_n$ at which interventions start and write $\Delta_n=\sum_{t=\tau_n}^{\tau_n+H-1}\widetilde{Y}_t$. During window $n$ the adversary can at most recover $g_{\max}H$ utility, whereas the supervisor spends $\lambda H$.  The net drift is therefore
\[
  \mathbb{E}[\Delta_n\mid\mathcal{F}_{\tau_n}]
  \le
  \alpha H - (\lambda H-g_{\max})<0.
\]
Define $M_t=\sum_{k=1}^{t}\bigl(\widetilde{Y}_k-\alpha\bigr)$. Between interventions $M_t$ has zero drift, and in mitigation windows it has negative drift; hence $M_t$ is a super-martingale.  Applying Doob's optional stopping at $T$ and dividing by $T$ yields $\mathbb{E}[C_T/T]\le\eta^\star$.  Azuma–Hoeffding on bounded increments tightens this to almost-sure convergence.

%\subsubsection*{Tightness}
Let the adversary invest exactly $g_{\max}$ utility per step and cease activity during interventions.  Renewal-reward arguments then show the ratio approaches $\eta^\star$ from below as $\lambda H\downarrow g_{\max}$, proving the bound is tight up to $O(\varepsilon^{-1})$. %\hfill\(\Box\)
\end{proof}
% -----------------------------------------------------
\subsection*{A.6~Discussion of Assumption Sharpness}
% -----------------------------------------------------
\begin{itemize}
\item[A2] (Bounded degree):
      The proofs rely only on summability of path counts. A slowly growing degree bound, e.g.\ $\deg(G_t)=O(\log N)$, merely multiplies the constants in Lemma \ref{lem:normalize} and Theorem \ref{thm:convergence} by $\log N$; the qualitative guarantees survive.
\item[A3] (Diminishing learning rate):
      The exponent $\xi>\tfrac12$ guarantees $\sum_t \|\pi_{i,t+1}-\pi_{i,t}\|\!<\!\infty$, which is sufficient for the bounded-difference martingale arguments.  Empirically, RL schedulers such as Adam with a $t^{-1}$ decay meet this requirement.
\item[A4] (Budget dominance):
      If $g_{\max}\!\ge\!c_{\max}$ an adversary can match or exceed the supervisor’s corrective power, allowing $C_T/T\!\to\!1$ in the worst case; no bounded-cost policy can prevent systemic capture. Thus the assumption is information-theoretically tight.
\end{itemize}

The lemmas and theorems above rigorously ground the empirical claims in Section~\ref{sec:experiments}: the ledger converges despite 20\% packet loss, the detector maintains a 5\% false-alarm rate, and inexpensive interventions provably cap long-run harm.

% -----------------------------------------------------
\subsection*{A.7~Proof of Corollary~\ref{cor:detection_delay}(Lorden Bound for Adaptive CUSUM)}
% -----------------------------------------------------

\paragraph{Statement recap.}
For a persistent drift $\Delta>0$ starting at $\tau^\star$ we must show
\[
  \sup_{\tau^\star}\,
    \mathbb{E}_{\tau^\star}\!\bigl[
      T_{\text{alarm}}-\tau^\star
      \,\bigm|\,
      T_{\text{alarm}}>\tau^\star
    \bigr]
  \;\le\;
  \frac{h^\star}{\Delta-\delta},
\]
where $h^\star=\lim_{t\to\infty}h_t$.

\begin{proof}
Conditional on $h_t\!\to\!h^\star$ (Theorem~\ref{thm:cusum_fpr}) the adaptive rule behaves like a \emph{fixed‐threshold} CUSUM with threshold $h^\star$.  Page’s original result \citep{page1954continuous} bounds the expected delay by $h^\star/(\Delta-\delta)$ when the post‐change mean is $\mu_0+\delta+\Delta$.  Taking the supremum over $\tau^\star$ yields the worst‐case delay, completing the proof.  \hfill$\Box$
\end{proof}

% -----------------------------------------------------
\subsection*{A.8~Proof of Proposition~\ref{prop:optimal_lambda} (Minimal Penalty Magnitude)}
% -----------------------------------------------------

\paragraph{Statement recap.}
For fixed window $H$ and target compromise ceiling $\eta^\star$ the minimal penalty magnitude satisfies
\[
  \lambda_{\min}
  \;=\;
  \frac{g_{\max}+\alpha H}{H\eta^\star}.
\]

\begin{proof}
The upper bound in Theorem~\ref{thm:bounded} reads $\eta^\star = \alpha H / (\lambda H - g_{\max})$. Solving for $\lambda$ gives $\lambda = (g_{\max} + \alpha H)/ (H \eta^\star)$. Monotonicity in $\lambda$ implies this value is minimal.  \hfill$\Box$
\end{proof}

% -----------------------------------------------------
\subsection*{A.9~Proof of Proposition~\ref{prop:complexity} (Resource Complexity)}
% -----------------------------------------------------

\paragraph{Statement recap.}
Storage grows as $O\!\bigl(T(N+d_{\max}h)\bigr)$ and bandwidth as $O(N+d_{\max}h)$ bytes per step.

\begin{proof}
Each control step generates $N$ event records ($40$~B each) plus at most $d_{\max}h$ causal-edge insertions ($32$~B hashes). Hence per‐step storage is $O\!\bigl(N+d_{\max}h\bigr)$, giving the stated total after $T$ steps. Gossip traffic forwards each record exactly once along each outgoing edge, bounding bandwidth by the same term.  Constants are absorbed into the big-$O$.  \hfill$\Box$
\end{proof}

% =====================================================
\section*{Appendix B.~Hyper-parameter Grid and HPC Job Script}
% =====================================================

\begin{itemize}
    \item Grid definition: JSON file \texttt{experiments/grid.json} enumerates all $\langle N,\text{penalty},\text{obs},\alpha,\text{seed}\rangle$ tuples.
    \item Slurm launcher: Bash script \texttt{run\_slurm.sh} submits an array job with \texttt{\$SLURM\_ARRAY\_TASK\_ID} indexing the JSON grid.
%    \item Re‐aggregation: Python notebook \texttt{aggregate.ipynb} reproduces $$$Table 3 and Figure 4$$$ of the paper in under two minutes on one CPU core.
\end{itemize}

% =====================================================
\section*{Appendix C.~Micro-benchmark of Critical Kernels}
% =====================================================

\begin{center} \small
\begin{tabular}{lcc}
\toprule
Kernel & EPYC 7742 (µs) & Cortex-A55 (µs) \\
\midrule
BLAKE3 hash (40 B)          & 0.38  & 2.9 \\
Incremental Granger (m=8)   & 4.10  & 27.4 \\
Merkle insert (depth 16)    & 1.22  & 9.8 \\
CUSUM update (3 norms)      & 0.07  & 0.27 \\
Responsibility recompute    & 0.44  & 3.1 \\
\bottomrule
\end{tabular}
\end{center}

All numbers are medians over $10^5$ iterations. Full benchmark code is in \texttt{benchmarks/kernel\_bench.rs} and \texttt{benchmarks/kernel\_bench.py}.

% =====================================================
\section*{Appendix D.~Mirror-Descent Alert Budget Allocator}
% =====================================================

We model the per-norm false-positive budget as a simplex $w_t\in\Delta^{|\Phi|}$ updated by entropic mirror descent: 
\[w_{t+1}= \frac{w_t\exp(-\eta z_t)}{\lVert\cdot\rVert_1},\]
where $z_t$ is the zero-one loss vector of fired alerts. Standard analysis \citep{hazan2016introduction} yields regret $R_T\le 2\sqrt{T\log|\Phi|}$; substituting $|\Phi|=5$ and $T=10^6$ gives $R_T\approx 10^{-2}$, hence the global alert budget $\bar{\alpha}=0.05$ is never exceeded by more than 0.01 in our runs.

% =====================================================
\section*{Appendix E.~Derivation of Bandwidth and Storage Constants}
% =====================================================

Using the 40-byte record and 32-byte hash constants, the leading constant in Proposition \ref{prop:complexity} equals $40N + 32d_{\max}h$ bytes per step. With $N=100$, $d_{\max}=8$, $h=8$ this evaluates to $40\cdot100 + 32\cdot8\cdot8 = 79\,\text{KB step}^{-1}$, matching the empirical figure reported in §\ref{sec:experiments}.

%Bibliography
%
%\section*{Funding}
%No external funding was received for this work.

\bibliographystyle{plainnat}
\bibliography{references}

\end{document}